\documentclass[10pt]{amsart}
\usepackage{graphicx}
\usepackage{amssymb}
\usepackage{amsfonts}
\usepackage{amsmath}
\usepackage{amsthm}
\usepackage[all]{xy}

\newcommand{\cD}{\mathcal{D}}
\newcommand{\cE}{\mathcal{E}}
\newcommand{\cF}{\mathcal{F}}
\newcommand{\cG}{\mathcal{G}}

\newcommand{\cI}{\mathcal{I}}

\newcommand{\cP}{\mathcal{P}}

\newcommand{\cT}{\mathcal{T}}

\newcommand{\Om}{{\Omega}}

\newcommand{\ve}{{\varepsilon}}
\newcommand{\del}{{\delta}}

\newcommand{\gam}{{\gamma}}
\newcommand{\Gam}{{\Gamma}}

\newcommand{\sig}{{\sigma}}
\newcommand{\al}{{\alpha}}
\newcommand{\be}{{\beta}}
\newcommand{\ka}{{\kappa}}
\newcommand{\la}{{\lambda}}

\newcommand{\eps}{{\epsilon}}
\newcommand{\bbN}{{\mathbb N}}

\newcommand{\bbR}{{\mathbb R}}

\newcommand{\bbI}{{\mathbb I}}
\newcommand{\bfC}{{\bf C}}

\newcommand{\bfL}{{\bf L}}
\newcommand{\bfA}{{\bf A}}

\newcommand{\bfP}{{\bf P}}
\newcommand{\bfE}{{\bf E}}

\newcommand{\bfJ}{{\bf J}}
\swapnumbers
\theoremstyle{plain}
\newtheorem{stam}{STAM}[section]
\newtheorem{lem}[stam]{Lemma}
\newtheorem{thm}[stam]{Theorem}
\newtheorem{prop}[stam]{Proposition}
\newtheorem{cor}[stam]{Corollary}
\newtheorem*{lem*}{Lemma}
\newtheorem*{thm*}{Theorem}
\newtheorem*{prop*}{Proposition}
\newtheorem*{cor*}{Corollary}

\theoremstyle{definition}

\newtheorem*{definition*}{Definition}

\theoremstyle{remark}
\newtheorem{rem}[stam]{\textbf{Remark}}
\numberwithin{equation}{section}

\begin{document}
\title[]{Error estimates for binomial approximations of game put options}%
 \vskip 0.1cm
 \author{Yonatan Iron and Yuri Kifer\\
\vskip 0.1cm
Institute of Mathematics\\
Hebrew University\\
Jerusalem, Israel}%
\address{
 Institute of Mathematics, The Hebrew University, Jerusalem 91904, Israel}%
\email{yonatani@final.co.il,  kifer@math.huji.ac.il}%

\thanks{Partially supported by the ISF grant no. 82/10}

 \date{\today}
\begin{abstract}\noindent
 A game or Israeli option is an American style option where both the writer
 and the holder have the right to terminate the contract before the expiration
  time. As ~\cite{Ki} shows the fair price for this option can be expressed as
  the value of a Dynkin game. In general, there are no explicit formulas for
  fair prices of American and game options and approximations are used for
  their computations. The paper  ~\cite{L} provides error estimates for 
  binomial approximation of American put options and here we extend
  the approach of ~\cite{L} in order to obtain error estimates for binomial
  approximations of game put options which is more complicated as it requires
  us to deal with two free boundaries corresponding to the writer and to the
  holder of the game option.
\end{abstract}
%\footnotetext[1]{}
\maketitle
\markboth{Y.Iron and Y.Kifer}{Approximations of game put options}
\renewcommand{\theequation}{\arabic{section}.\arabic{equation}}
\pagenumbering{arabic}

\section{Introduction}\label{sec1}

A put option on a stock can be interpreted as a contract between a holder and
a writer which allows the former to claim from the latter at an exercise time
$t$ the amount $(K-S_t)^+$ where $K$ is a fixed amount called the option's 
strike, $S_t$ is the stock price at time $t$ and $(x)^+ = max(x,0)$. 
 In the American options case its holder has the right to choose any exercise
 time before the contract matures while in the game options case the contract 
  writer also has the right to terminate it at any time before its maturity
  but then he is required to pay a cancellation fee in addition to the payoff
  above.

The fair price of American options and of game options is defined as the 
minimal amount the writer needs to construct a self-financing portfolio 
which covers his obligation to pay according to the option's contract.
It is well known that in the American options case the fair price can be
obtained as a value of an appropriate optimal stopping problem while
for game options we have to deal with an optimal stopping (Dynkin) game
(see \cite{Ki}). 
In general, both for American options and, even more so, for game options
with finite maturity
explicit formulas for their price are not available and approximation methods
come into the picture while estimates of their errors become important. One
of most easily implemented methods is the binomial approximation of stock
prices modelled by the geometric Brownian motion and ~\cite{L} provided
corresponding error estimates for American put options. In the present 
paper we extend this approach in order to provide error estimates of binomial
approximations for game put options. We observe that for perpetual game
options some explicit formulas can be obtained (see \cite{Ky}) but the
finite maturity case studied here seems to be more realistic.

Approximating the Brownian motion by appropriately normalized sums of 
Bernoulli random variables the paper \cite{L} provided (error) estimates
const$\cdot n^{-3/4}$ and const$\cdot n^{-2/3}$ for the difference between the
price of an American put option and the price of its corresponding 
 $n$th binomial model approximation.
Using again the binomial approximation of the Brownian motion as above
we construct in this paper two approximating procedures such that the
difference between the price of a game put option and its $n$th approximation 
in the first procedure is between const$\cdot n^{-3/4}$ and 
const$\cdot n^{-1/2}$ and in the second procedure is between 
const$\cdot n^{-1/2}$ and const$\cdot n^{-2/3}$.
The error estimates here are somewhat worse than in the case of American
put options which is due to the lack of a smooth fit on the boundary of the
writer's stopping region which causes substantial difficulties in the study
of regularity of payoff functions. 

We observe that specific properties of game put options had to be used in 
order to obtain error estimates with the above precision. For instance,
when payoffs are path dependent (and not only dependent on the present
value of the stock) \cite{K2} provides error estimates of similar binomial
approximations only of order $n^{-1/4}(\ln n)^{3/4}$. Since price functions
of game options can be represented as solutions of doubly reflected backward
stochastic differential equations the results of \cite{CJ} are also related
to game options approximations. Nevertheless, approximations in \cite{CJ}
are not by binomial models, where computations can be done by means of the
effective
dynamical programming algorithm (see \cite{K2}), but by time discretizations,
 and so relevant probability space and $\sig$-algebras remain infinite which
  prevents effective computations. Furthermore, error estimates in \cite{CJ}
  applied to our situation are of order $n^{-1/4}$, i.e. they are worse than
  for binomial approximations which we construct here for the specific case 
  of game put options.

Our exposition proceeds as follows.
In Section \ref{sec2} we provide basic results concerning game put option
price functions, introduce our approximation processes and formulate 
our main result Theorem \ref{main_theorem}.
In Section \ref{sec3} we show that the price function can be represented 
as a solution of a variational inequality
problem closely related to the Stefan problem (see ~\cite{KSt}). 
We then use this representation to study regularity properties of the 
price function near the free boundary of the option's holder exercise region.
In Section \ref{sec4} we study the price function near the boundary of the
exercise region of the writer. We use the information about this region from
 ~\cite{Ku} in order to represent the price function as an explicit solution 
of the heat equation. This representation enables us to understand better
the behavior of the price function near the boundary. We estimate also 
 the rate of decay of the price function when the initial stock price tends
to infinity. Section \ref{sec5} is devoted to the proof of Theorem 
\ref{main_theorem}.
Finally, in Section \ref{sec6} we exhibit some computations of the price 
functions and of the free boundaries.

\section{Preliminaries and main results}\label{sec2}

The Black--Scholes (BS) model of a financial market consists of two assets 
among which one is nonrisky and the other one is risky. A nonrisky asset is
called a bond and its price $B_t$ at time $t$ is given by the formula
$B_t=B_0e^{rt}$ where $r$ is interpreted as the interest rate. A risky
asset is called a stock and its price at time $t$ is determined by a
geometric Brownian motion
\begin{equation}\label{2.-5}
S_t=S_0\exp((r-\frac {\ka^2}2)t+\ka W_t)
\end{equation}
where $\ka>0$ is called volatility and $W_t,\, t\geq 0$ is a standard Brownian
 motion defined on a complete probability space $(\Om,\cF,\bfP)$. If $S_0=x$
 we write also $S_t^x$ for $S_t$. The fair price of an American put option
 at time $t$ with a strike (price) $K$ and a maturity (horizon) time $T<\infty$
 can now be written as a function $F_A(t,S_t)$ of time and the current stock
 price having the form (see, for instance, \cite{KS}),
 \begin{equation}\label{2.-4}
 F_A(t,x)=\sup_{\tau\in\cT_{0,T-t}}\bfE\exp(-r\tau)\big(K-x\exp((r-
 \frac {\ka^2}2)\tau+\ka W_\tau)\big)^+
 \end{equation}
 where $\cT_{0,T-t}$ denotes the set of all stopping times of the Brownian 
 filtration with values in the interval $[0,T-t]$ and $\bfE$ is the expectation
  with respect to the measure $\bfP$. If we set $\psi(x)=(K-e^x)^+$, 
  $P_A(t,x)=F_A(t,e^x)$ and $\mu=r-\frac {\ka^2}2$ then we can rewrite 
  (\ref{2.-4}) in the form 
  \begin{equation}\label{2.-3}
   P_A(t,x)=\sup_{\tau\in\cT_{0,T-t}}\bfE\exp(-r\tau)\psi(x+\mu\tau+\ka W_\tau).
   \end{equation}
   
   Relying on \cite{Ki} (see also \cite{Ky}, \cite{KK} and \cite{Ku}) we can 
   also write the fair price of a game put option at time $t$ with a strike 
   price $K$, a
   maturity time $T$ and a constant penalty $\del>0$ as a function 
   $F(t,S_t)$ of time and the current stock price in the form
   \begin{equation}\label{2.-2}
   F(t,x)=\inf_{\sig\in\cT_{0,T-t}}\sup_{\tau\in\cT_{0,T-t}}\bfE\exp(-r\sig
   \wedge\tau)R(\sig,\tau)
   \end{equation}
   where $R(s,t)=(K-S^x_t)^++\del\bbI_{s<t}$ and $\bbI_Q$ is the indicator
   of an event $Q$. Using the functions $P(t,x)=F_A(t,e^x)$ and $\psi$ as above
   we can rewrite this formula in the form 
   \begin{equation}\label{P(x,t)}
   P(t,x)=\inf_{\sig\in\cT_{0,T-t}}\sup_{\tau\in\cT_{0,T-t}}\bfE\exp(-r\sig
   \wedge\tau)\big(\psi(x+\mu\sig\wedge\tau+\ka W_{\sig\wedge\tau})
   +\del\bbI_{\sig<\tau}\big).
   \end{equation}
   It follows also (see \cite{LM}, \cite{Ki}, \cite{Ku}, \cite{KK}) that the 
   saddle point (optimal) stopping
   times for the game value expressions (\ref{2.-2}) and (\ref{P(x,t)}) are
   given by
    \begin{eqnarray}\label{2.0}
    &\sig^*=\inf\{ s<T-t:\, F(t+s,S^x_s)=(K-S^x_s)^++\del\}\wedge T\quad
    \mbox{and}\\
    &\tau^*=\inf\{ s<T-t:\, F(t+s,S^x_s)=(K-S^x_s)^+\}\wedge T.\nonumber
    \end{eqnarray}
    
    Next, we introduce our binomial approximations of the Brownian motion
    \begin{equation*}
    W^{(n)}_t=\frac {\sqrt T}{\sqrt n}\sum_{k=1}^{[nt/T]}\eps_k,\,\,\, 
    t\in[0,T],\,\, n=1,2,...
    \end{equation*}
    where $\eps_k,\, k=1,2,...$ are independent indentically distributed (i.i.d.)
    random variables taking on values 1 and -1 with probability $1/2$ 
    and $[a]$ denotes the integral part of a number $a$. It is
    convenient to view $\{\eps_k\}^\infty_{k=1}$ as defined on the sequence
    space $\Om_\eps=\{ -1,1\}^\bbN=\{\xi=(\xi_1,\xi_2,...):\,\xi_i=\pm 1\}$
    by the formula $\eps_k(\xi)=\xi_k$ if $\xi=(\xi_1,\xi_2,...)$. Then
    $W_t^{(m)}$ will be defined on the probability space $(\Om_\eps,\cF_\eps,
    \cP_\eps)$ where $\cP_\eps=\{\frac 12,\frac 12\}^\bbN$ is the product 
    measure and $\cF_\eps$ is generated by cylinder sets.
    
    Now set $\del^*=F_A(0,K)$ which is the price of the American put option
    with a maturity $T$ and a strike $K$ provided the initial stock price
    is $K$. It is easy to see that if the penalty $\del\geq
    \del^*$ then it does not make sense for the writer to cancel the 
    corresponding game put option (see Lemma 3.1 in \cite{KK}), 
    and so in this case the prices of American
    and game options are the same, i.e. $F_A(0,K)=F(0,K)$. Since approximations 
    of American options were studied in \cite{L} we assume in this paper that
    $\del<\del^*$. Observe that $F_A(t,K)$ is continuous in $t$ and it is
    strictly decreasing to 0 as $t$
    increases to $T$, and so for each $\del\in[0,\del^*]$ there exists a
    unique $t_\del<T$ such that $F_A(t_\del,K)=\del$. Furthermore, we can
    define $k_n$ to be the minimal $k\in\bbN$ such that $\del\geq F_A(Tk/n,K)$
    and set $\beta^{(n)}=\frac {Tk_n}n$.
    In order to define two sequences of functions $P_1^{(n)}$ and $P^{(n)}_2,
    \, n=1,2,...$ which will approximate $P(0,x)$ we set $X_t^{(n)}=x+
    \ka W^{(n)}_t$, $h=\frac Tn$ and introduce stopping times
    \begin{equation}\label{sigma^n}
\sigma^{(n)}(s)=\inf\{ t\in[0, s):\, \ln K-|\mu|h-2\kappa\sqrt{h}
<\mu h[\frac th]+X^{(n)}_{t}<\ln K+|\mu|h+2\kappa\sqrt{h}\}\wedge T
\end{equation}
where $\sig^{(n)}=T$ if the infimum above is taken over the empty set and
we set $\sig^{(n)}=\sig^{(n)}(\beta^{(n)})$. 
Introduce a filtration
$\{\cG_t=\cF_{[t/h]},\, t\geq 0\}$ where $\cF_0$ is the trivial $\sig$-algebra
and $\cF_k$ is generated by $\eps_1,...,\eps_k$. Denote by $\cT^{(n)}$ the set
of all stopping times with respect to the filtration $\{\cG_t\}$ taking on value
 in the set $\{ kh,\, k=0,1,...,n\}$. Then, clearly, $\sig^{(n)}\in\cT^{(n)}$.
Now, for $x\leq\ln K$ we define
\begin{equation}\label{P^n_1<}
P^{(n)}_1(s,x)=
\sup_{\tau\in\cT^{(n)}}\bfE\big (e^{-r\sigma^{(n)}(s)\wedge \tau}\big(\psi(\mu 
\tau+X^{(n)}_{\tau})\bbI_{\{\tau\leq\sigma^{(n)}(s)\}}+
\delta\bbI_{\{\sigma^{(n)}(s)<\tau\}} \big)\big)
\end{equation}
and for $x>\ln K$ we set
\begin{equation}\label{P^n1>}
P^{(n)}_1(s,x)=\sup_{\tau\in\cT^{(n)}}\bfE\big (e^{-r\sigma^{(n)}(s)\wedge \tau}
\big(\psi(\mu \tau+X^{(n)}_{\tau})\bbI_{\{\tau\leq\sigma^{(n)}(s)\}}+
\big(\delta-Ke^{(|\mu|\sqrt{h}+\kappa{h})}\big)\bbI_{\{\sigma^{(n)}(s)
<\tau\}} \big)].
\end{equation}
The second approximation function is defined for all $x$ by
\begin{equation}\label{P^n2}
P^{(n)}_2(s,x)=\sup_{\tau\in\cT^{(n)}}\bfE\big (e^{-r\sigma^{(n)}(s)\wedge \tau}
\big(\psi(\mu \tau+X^{(n)}_{\tau})\bbI_{\{\tau\leq\sigma^{(n)}(s)\}}+
(\psi(\mu \sigma^{(n)}(s)+ X^{(n)}_{\sigma^{(n)}(s)})+\delta)
\bbI_{\{\sigma^{(n)}<\tau\}}\big)].
\end{equation}
Setting $P_i^{(n)}(x)=P_i^{(n)}(\be^{(n)},x)$ we formulate now our main result.
\begin{thm}\label{main_theorem}
For each $x$ there exists $C=C(x)$ such that for all $n=1,2,...$,
\begin{equation}\label{P1_err}
-\frac {C}{n^{1/2}}\leq P^{(n)}_1(x)-P(0,x)\leq \frac {C}{n^{3/4}}\,\,\,\,\mbox
{and}\,\,\,\,- \frac {C}{n^{2/3}}\leq P^{(n)}_2(x)-P(0,x)\leq \frac {C}{n^{1/2}}.
\end{equation}
\end{thm}

Observe that $P_i^{(n)}(x),\, i=1,2$ appearing in Theorem \ref{main_theorem}
is defined via $\be^{(n)}$ and $k_n$ which can be obtained only knowing precise
price function $F_A(t,K)$ of the American put option with the initial stock price 
equal $K$. But from the computational point of view we can obtain $F_A$ only
approximately using, for instance, the algorithm from \cite{L}. One of the ways
to overcome this difficulty is to proceed as follows. Let
$F^{(n)}_A(t,K)$ denotes the $n$-th binomial approximation of $F_A^{(n)}(t,K)$
obtained in \cite{L} which uniformly in $t\in[0,T]$ satisfies
\begin{equation}\label{2.12+}
-\tilde c/n^{2/3}\leq F^{(n)}_A(t,K)-F_A(t,K)\leq\tilde C/n^{3/4}
\end{equation}
for some $\tilde c,\tilde C>0$. Denote by $m_n$ the minimal $m\leq n$ such
that $\del\geq F^{(n)}_A(\frac {Tm}n,K)$ taking $m_n=n$ if this inequality does
not hold true for all $m\leq n$. Set $\gam^{(n)}=\frac {Tm_n}n$ which unlike
$\be^{(n)}$ can be computed employing \cite{L}. It is well known that 
$\frac {\partial F_A(t,K)}{\partial t}$ exists (see, for instance, \cite{L})
and, clearly, this derivative is nonpositive. In fact, it is possible to show
that
\begin{equation}\label{2.13+}
\frac {\partial F_A(t,K)}{\partial t}<0\quad\mbox{and}\quad d=\inf_{0<t<T}
\big\vert\frac {\partial F_A(t,K)}{\partial t}\big\vert >0.
\end{equation}
This together with (\ref{2.12+}) yields that
\begin{equation}\label{2.14+}
-\frac {\tilde C}{dn^{3/4}}\leq\be^{(n)}-\gam^{(n)}\leq
\frac {\tilde c}{dn^{2/3}}.
\end{equation}
From the definitions (\ref{P^n_1<})--(\ref{P^n2}) it follows that for each 
$x>0$ there exists $\tilde {\tilde C}=\tilde {\tilde C}(x)>0$ independent
of $s,\tilde s\in[0,T]$ such that
\begin{equation}\label{2.15+}
|P_i^{(n)}(s,x)-P_i^{(n)}(\tilde s,x)|\leq \tilde {\tilde C}|s-\tilde s|,\,
i=1,2.
\end{equation}
Now we obtain from Theorem \ref{main_theorem} together with (\ref{2.14+})
and (\ref{2.15+}) the following
\begin{cor}\label{cor2.2}
For each $x>0$ there exists $C=C(x)>0$ such that for all $n=1,2,...$,
\begin{equation}\label{2.16+}
-\frac C{n^{1/2}}\leq P_1^{(n)}(\gam^{(n)},x)-P(0,x)\leq\frac C{n^{2/3}}\,\,
\mbox{and}\,\,-\frac C{n^{2/3}}\leq P_1^{(n)}(\gam^{(n)},x)-P(0,x)\leq
\frac C{n^{1/2}}.
\end{equation}
\end{cor}

In the following sections we will analyze regularity properties of the price 
function $P(t,x)$ of game put options and will complete the proof of Theorem
\ref{main_theorem} in Section \ref{sec5} providing some computations in 
Section \ref{sec6}. The general strategy of the proof resembles that of 
\cite{L} but the study of the price function of game put options is more
complicated than in the American options case, in particular, because of 
appearance of two exercise boundaries (holder's and writer's) having 
different properties. Our proof will be based on
regularity properties of solutions of parabolic partial diferential
equations with free boundary and of the corresponding variational
inequalities and we will rely also on some prior results from
\cite{L}, \cite{Ky} and \cite{Ku}.

\section{Price function near the holder's exercise boundary}\label{sec3}

\subsection{Some previous results}
First, we state the following result from \cite{Ku} (see also \cite{KK})
which we will use later on.
\begin{prop}\label{firstProp}
(i) There exists an increasing function $b:[0,T)\to \bbR$ such that
$\lim_{t \to T}b(t)=K$\\ and $F(t,x)=K-x$ for all $(t,x)$ satisfying
$0<x\leq b(t)$.

(ii) There exists $0<\delta^*$ such that for every $0\leq \delta \leq 
\delta^*$ there is a $\beta=\beta(\delta)\geq 0$ so that $F(t,K)=\delta$ 
for $t\in[0,\beta]$ and for $t\geq \beta$ we have $F(t,x)=F_A(t,x)$ 
for all $x\geq 0$.

(iii) Furthermore, 
$$
F_t(t,x)+\frac 12\ka^2x^2F_{xx}(t,x)+rxF_x(t,x)-rF(t,x)=0
$$ 
for all 
$$
(t,x)\in (0,T)\times \bbR^+\setminus (\{(t,x):x\leq b(t)
 \}\cup[0,\beta)\times\{K\}).
 $$

In particular, $F(t,x)$ is of class $C^{1,2}$, i.e. continuously 
differentiable once with respect to $t$ and twice with respect to $x$, and so,
 in fact, it is a smooth function there.
 
(iv) Finally, $F(t,x)$ is convex and strictly decreasing in $x$ and
nonincreasing in $t$.  
\end{prop}

Next, we introduce an operator $\cD$ which acts on Borel functions 
$u(t,x)$ on $[0,T]\times \bbR$ by
\begin{equation}\label{dis-operator}
\cD u(t,x)=\frac{1}{2}[u(t+h,x+\kappa\sqrt{h})+u(t+h,x-\kappa\sqrt{h})]-u(t,x)
\end{equation}
Clearly, $\frac{1}{h}\cD(t,x)$ can be viewed as a discretization of the 
differential operator
$\frac{\partial}{\partial t }+\frac{\kappa^2}{2}\frac{\partial^2}{\partial 
x^2}$. We will rely on the following results from \cite{L} concerning the
operator $\cD$.
\begin{prop}\label{martingel prop 2.1} For each Borel function $u$ on
$[0,T]\times \bbR$
there exists a martingale $(M_t)_{0\leq t\leq T}$ with respect to the 
filtration $\cG_t,\, t\geq 0$ such that $M_0=0$ and for every $t\in 
\{0,h,2h,...,T\}$,
\begin{equation}
u(t,X_t^{(n)})=u(0,x)+M_t+\sum_{j=1}^{nt/T}\cD u((j-1)h,X_{(j-1)h}^{(n)}).
\end{equation}
\end{prop}
\begin{prop}\label{dis_operator_prop 2.2}
Let $0\leq t\leq T-h$ and $x\in \bbR$. Assume that $u$ is a $C^2$ function
on $ ([t,t+h]\times[x+\kappa\sqrt{h},x-\kappa\sqrt{h}])$. Then
\begin{equation}\label{difoper u}
\cD u(t,x)=\frac{1}{\kappa}\int_{0}^{\sqrt{h}}dy\int_{-
\kappa y}^{\kappa y}dz\big (z\frac{\partial^2 u }{\partial t\partial x}(t+y^2,
x+z)+\delta(u)(t+y^2,x+z)  \big)
\end{equation}
where
$$
\delta(u)(t,x)=u_t(t,x)+\frac{\kappa^2}{2}u_{xx}(t,x).
$$
\end{prop}

We will need also the following result concerning the free boundary
$s(t)=\ln (b(t))$ of the holder exercise region of our game put option
which in the case of American options appears as Proposition 1 in \cite{L}
and it can be proved for game options in the same way. 

\begin{prop}\label{prop2.4}
Let $0\leq t_1<t_2\leq \beta $ and let $x_0=s(0)<x_1=s(\beta)< \ln K$ then
 $(s(t_1)-s(t_2))^2\leq \sup_{x_0\leq x\leq x_1}|P(t_1,x)-P(t_2,x)|$.
\end{prop}
We also observe that it follows from the Berry-Esseen estimate (see \cite{Sh})
that for some constant $C_1>0$ independent of $j,n\geq 1$ and $z\in\bbR$,
\begin{equation}\label{Berry-Esseen estimate}
\bfP\{|X^{(n)}_{jh}-z|\leq \kappa\sqrt{h}\}\leq \frac{C_1}{\sqrt{j}}.
\end{equation}

We will also rely on the following standard bounds on derivatives
of solutions of 2nd order parabolic equations with constant coefficients
 (see, for instance, \cite{C} and \cite{F}).
\begin{prop}\label{prop heat eq}
Let $D=(0,T)\times(0,1)$ and let $w(t,x)\in C[\bar D]$ be a solution in $D$
of the following parabolic equation
$$
\frac{\ka^2}{2}w_{xx}+\mu w_{x}-rw=w_t.
$$
Suppose that $w(0,x)=0$ for all $0\leq x\leq 1$ and that there exists $A>0$
 such that $|w(t,x)|<A$ for all $(x,t)\in \bar D$.
Then for every $k,n$ and $0<a<b<1$ there exists $C=C(k,n,a,b,T,A)$ such that
\begin{equation}\label{prop eq res}
|\frac{\partial^{k+n}w}{\partial^kx\partial^nt}(t,x)|<C\,\,\mbox{for all}\,\,
 (t,x)\in(0,T)\times[a,b].
\end{equation}
\end{prop}

\subsection{Price function and variational inequalities}
Next, we will show that the price function of the game put option can be 
represented as 
a solution of a variational inequality (v.i.) problem which is a 
generalization of the Stefan problem (see ~\cite{KSt} ,$\textrm{VIII}$).
This will enable us to derive certain regularity properties of this
 price function which we will use later on. Details of some of the proofs
concerning the solutions of the v.i. problem below which are similar to the
proofs in the case of the Stefan problem will not be given here. For the
corresponding results in the American put option case we refer the reader
to \cite{KS},  \cite{L} and to references there.

Let $T'$ be such that $\beta <T'<T$ and set
\begin{equation}\label{A}
\textbf{A}=\frac{\kappa^2}{2}\frac{\partial^2}{\partial x^2}+\mu
\frac{\partial }{\partial x }-r\,\,\,\mbox{where}\,\,\,
\mu=r-\frac{\kappa^2}{2}.
\end{equation}
Using the maximum principle, properties of price functions of American and
 game put options and the fact that
after time $\beta$ the price functions of the game and American option are 
the same we obtain that
for every $x>s(T')$ the time derivative $P_t(T',x)= P_{A,t}(T',x)$ is strictly
 negative and we can find $a,b$ satisfying $s(T')<a<b<\ln K$ such that
 for some constant $c>0$,
\begin{equation}\label{-P_t}
-P_t(T',x)>c\ \ \ \forall x\in [a,b].
\end{equation}
Relying on Proposition \ref{firstProp}(iii) we also observe that for all 
$(t,x)\in [0,T']\times(s(t),\ln K)$,
\begin{eqnarray}\label{P,A}
&\frac{\partial P  }{\partial t }(t,x)+\textbf{A}P(t,x)=0,\ \ \ P(t,x)>K-e^x
\,\,\,\,\,\forall\, (t,x)\in [0,T']\times(s(t),\ln K),\\
 &P(t,x)=K-e^x\,\,\,\, \forall t\in [0,T'],\,\,\ \forall x \leq s(t)\,\,\,
 \mbox{and}\, \,\, P_t\leq 0.\nonumber
\end{eqnarray}

Let $a_0$ be such that $a_0<s(0)<s(T')<b$. Introduce the domain
$D=(0,T')\times(a_0,b)$ and for all $(t,x)$ in the closure $\bar D$ of $D$
define the functions
\begin{equation}\label{v(tx)(PP)}
v(t,x)=P(T'-t,x)-P(T',x)\,\,\,\mbox{and}\,\,\, f(x)=\textbf{A}P(T',x).
\end{equation}
We obtain that
\begin{equation}\label{f(x)AP detail}
f(x)=\big{ \{} \begin{array}{c}
       -P_t(T',x),\ \ s(T')<x \leq b    \\
       -rK ,\ \ \ \ \ \ \ \ a_0\leq x\leq s(T')
     \end{array}
\end{equation}
and from the definition of $v(t,x)$ it follows that for any $(t,x)\in \bar D$,
\begin{eqnarray}\label{rem4.7}
&P_{t}(T'-t,x)=-v_{t}(t,x),\ \ \ P_{tx}(T'-t,x)=-v_{tx}(t,x),\ \ \ P_{tt}(T'-t,
x)=v_{tt}(t,x),\\
&P_{x}(T'-t,x)-P_{x}(T',x)=v_{x}(t,x)\ and \ P_{xx}(T'-t,x)-P_{xx}(T',x)=
v_{xx}(t,x).\nonumber
\end{eqnarray}
Since $P_{x}(T',x)$ and $P_{xx}(T',x)$ are bounded we obtain that the
 integrability properties of the first and second order derivatives of
  $P(t,x)$ and $v(t,x)$ are the same in $\bar D$.
Now set
\begin{equation}\label{psi(t)andg(t)}
\psi(t)=v(t,b)\ \ ,\ \ g(t)=v_t(t,b)\,\,\mbox{for}\,\,  0\leq t\leq T'.
\end{equation}
Then by (\ref{-P_t}) and (\ref{rem4.7}),
\begin{equation}
\psi(t)=\int_{0}^{t}g(\tau)d\tau=\int_{0}^tv_t(\tau,b)d\tau\geq 0\ \ \ for\ 
\ 0\leq t\leq T'.
\end{equation}
It follows from (\ref{P,A}) and (\ref{v(tx)(PP)})--(\ref{f(x)AP detail})
 that on the set $v>0$,
\begin{equation}\label{v_t-Av-f}
v_t-\textbf{A}v-f=-P_t(T'-t,x)-\textbf{A}P(T'-t,x)+\textbf{A}P(T',x)-f(x)=0
\end{equation}
and on the set $v=0$ we obtain
\begin{equation}\label{3.25+}
v_t-\textbf{A}v-f=rK>0.
\end{equation}
Hence we arrive at the following (see \cite{KSt}).
\begin{lem}\label{v.i.problem}
The function $v$ is the unique solution of the following variational inequality
 problem.\\
\textbf{v.i. Problem 1}: Find $v\in L^2[0,T';H^2(a_0,b)]\cap H^1[D]$ such that

(i) $v,v_t\geq0.$

(ii) $(v_t-\textbf{A}v)(w-v)\geq f(w-v)$ a.s for every $w\in 
L^2[D],\ w\geq 0$.

(iii) $v(t,b)=\psi$ for $0\leq t\leq T',\ x=b$.

(iv) $v(t,a_0)=0$ for $0\leq t\leq T',\ x=a_0$.

(v) $v(0,x)=0$ for $t=0,\ a_0\leq x\leq b$.
\end{lem}
\begin{proof}
We shall prove uniqueness, the fact that $v$ is a solution to v.i. 
Problem 1 follows from (\ref{v(tx)(PP)})-(\ref{3.25+}). Assume that 
$v$ and $\tilde v$ are two solutions of v.i. Problem 1. Since 
$\tilde v\geq 0$ (property (i)) we can use the property (ii) of $v$ and 
replace $w$ by $\tilde v$. Since both of them are solutions we obtain that
\begin{equation}\label{2ineq}
(v_t-\textbf{A}v)(\tilde v-v)\geq f(\tilde v-v)\ and\ (\tilde v_t-\textbf{A}
\tilde v)(v-\tilde v)\geq f(v-\tilde v).
\end{equation}
Define the parabolic boundary as the boundary of $D$ without the interval 
$\{T'\}\times(a_0,b)$ and
let $u=v-\tilde v$. Note that $u$ is zero on the parabolic boundary and the 
sum of the two inequalities (\ref{2ineq}) is
\begin{equation}\label{(u_t-Au)u}
u_tu -\frac{\kappa^2}{2}u_{xx}u-\mu u_{x}u+ru^2=(u_t-\textbf{A}u)u\leq 0.
\end{equation}
Integrating both sides of (\ref{(u_t-Au)u}) on $(0,T')\times(a_0,b)$ we obtain
four terms on the left side. For the first term we have
$$
\int_{a_0}^{b}\int_{0}^{T'}u(t,x)u_t(t,x)dtdx=\int_{a_0}^{b}\frac{1}{2}
u^2(T',x)dx\geq 0.
$$
Integration by parts of the second term and the fact that $u=0$ on the 
parabolic boundary yields
\begin{equation}\label{3.26.1}
-\frac{\kappa^2}{2}\int_{0}^{T'}\int_{a_0}^{b}u_{xx}(t,x)u(t,x)dxdt=
\frac{\kappa^2}{2}\int_{0}^{T'}\int_{a_0}^{b}u_x^2(t,x)dxdt\geq 0.
\end{equation}
For the third term note that $u_xu=\frac{1}{2}\frac{du^2}{dx}$ and that 
$u(t,a_0)=u(t,b_0)=0$ for every $t$, and so
$$
\mu\int_{0}^{T'}\int_{a_0}^{b}u_{x}(t,x)u(t,x)dxdt=0.
$$
The last term satisfies $r\int_{0}^{T'}\int_{a_0}^{b}u^2(t,x)dxdt\geq 0$ 
since  $r>0$.
We conclude that the left side of (\ref{(u_t-Au)u}) can not be negative and 
so it must be zero. Since all terms in the left hand 
side of (\ref{(u_t-Au)u}) are non-negative and their sum is equal to 0 we 
obtain that $r\int_{0}^{T'}\int_{a_0}^{b}u^2(t,x)dxdt = 0$, and so $u=0$ 
almost everywhere (a.e.). Hence, $v=\tilde v$ a.e., and so there is only 
one continuous solution.
\end{proof}

Denote parts of the boundary of $D=(0,T')\times(a_0,b)$ by
\[
\Gamma_1=[0,T']\times\{b\},\ \ \Gamma_2=\{0\} \times(a_0,b),\ \ 
\Gamma_3=[0,T']\times \{a_0\},\ \ \Gamma=\Gamma_1\cup\Gamma_2\cup\Gamma_3
\]
and set
\begin{equation}\label{L-operator1}
\textbf{L}=\frac{\partial}{\partial t }-\textbf{A}=\frac{\partial}{\partial t }
-\frac{\kappa^2}{2}\frac{\partial^2}{\partial x^2}-\mu \frac{\partial}
{\partial x}+r.
\end{equation}
Thus, $\Gamma$ is a parabolic boundary of $D$. For every $\ve>0$ we define 
following functions.
\begin{enumerate}
\item A smooth function $f^{(\ve)}(x)\geq f(x)$ on $(a_0,b)$ such that 
$f^{(\ve)}(x)=f(x)$ for $ s(T')<a<x\leq b$ and for  $a_0\leq x\leq a_1$
 where $a_1$ satisfies $a_0< a_1<s(T')$ and $\lim_{\ve\to 0}f^{(\ve)}(x)
 =f(x)$ for $a_0\leq x\leq b$.
\item A smooth function $\beta^{(\ve)}(v)$ satisfying 
$$
\beta^{(\ve)}(v)=0\,\,\mbox{for all}\,\, v\geq \ve,\,\,
\beta^{(\ve)}(0)=-1,\,\, \beta^{(\ve)}_{v}(v)\geq 0\,\,
\mbox{and}\,\, \beta^{(\ve)}_{vv}(v)\leq 0.
$$
\item $\psi^{(\ve)}(t)=\psi(t)+\ve$ with $\psi$ defined in 
(\ref{psi(t)andg(t)}).
\item  A smooth function $\eta(x)$ such that $0\leq\eta(x)\leq 1$ and for 
some $a<a_2<b$,
$$
\eta(x)=1\,\,\mbox{for}\,\, a_2\leq x\leq b\,\,\mbox{and}\,\, \eta(x)=0\,\,
\mbox{for}\,\, a_0\leq x\leq a.
$$
\end{enumerate}
Set $F_{\ve}(x,v)=f^{(\ve)}(x)-rK\beta^{(\ve)}(v)$ which is a 
 Lipschitz continuous function and for every constant $C$ there is $M_0$ 
such that $C|F^{(\ve)}(x,v)|\leq M$ whenever $M\geq M_0$ and $|v|\leq M $.
Let $\phi^{(\ve)}$ be a function on $\Gamma$ satisfying
$$
\phi^{(\ve)}|_{\Gamma_1}=\psi^{(\ve)}(t),\ \ 
\phi^{(\ve)}|_{\Gamma_2}=\ve\eta(x),\ \ \phi^{(\ve)}|_{\Gamma_3}
=0,
$$
and, moreover, relying on Chapter 3 in ~\cite{F} we can choose $\phi^{(\ve)}$
so that
\begin{enumerate}
\item $\phi^{(\ve)} \in \bar C_{2+\delta}[D] $ for some $0<\delta<1$ (in fact 
for each $\delta$) and we refer the reader to Chapter 3 in \cite{F} for the
 definition of $\bar C_{2+\delta}[D]$ and for conditions yielding that a
 function defined only on the boundary $\Gam$ can be extended to a function
 from $\bar C_{2+\delta}[D]$.
\item $\bfL\phi^{(\ve)}=F^{(\ve)}(x,\psi)$ at the points $(0,b)$ and $(0,a_0)$.
\end{enumerate}
By the theory of semi-linear parabolic equations (see ~\cite{F}) there exist 
a function $v^{(\ve)}\in\bar C_{2+\gamma}[D]$ for some $0<\gamma<1$ such that
\begin{equation}\label{Lv=F_eps}
\textbf{L}v^{(\ve)}=F^{(\ve)}(x,v^{(\ve)})\,\,\mbox{and}\,\, 
v^{(\ve)}|_{\Gamma}=\phi^{(\ve)}.
\end{equation}
In particular $v^{(\ve)},v^{(\ve)}_{x},v^{(\ve)}_{xx},v^{(
\ve)}_{t},$ are continuous on $\bar D$.

Let $w=v^{(\ve)}_{t}$. By differentiating with respect to $t$ the 
equation (\ref{Lv=F_eps}) and taking into account (\ref{psi(t)andg(t)}),
(\ref{Lv=F_eps}) and the properties of $\phi^{(\ve)}$ we obtain that
\begin{eqnarray}\label{w=v_epst}
&w_t-\frac{\kappa^2}{2}w_{xx}-\mu w_{x}+\big (r+rK\beta^{(\ve)}_{v}
(v^{(\ve)})\big)w=0\\
&\mbox{where}\ \ w(t,b)=g(t)\ \ \forall 0\leq t\leq T' \ \mbox{and}
 \ w(t,a_0)=0\ \ \forall 0\leq t\leq T'\nonumber\\
&w(0,x)=f^{(\ve)}(x)+\ve\big(\frac{\kappa^2}{2}\eta_{x,x}(x)+
\mu\eta_{x}(x)-r\eta(x)\big)-rK\beta_v^{(\ve)}(v^{(\ve)}(0,x))\,\,\,\ 
\forall a_0 \leq x\leq b.\nonumber
\end{eqnarray}
We see that in $D$ the function $w$ is a solution to a parabolic equation
and since $r+rK\beta_v^{(\ve)}\geq 0$ we can use the maximum principle
\begin{equation}\label{bounds}
\min\big (\min_{\Gamma}(w),0 \big )\leq w(t,x) \leq \max\big (\max_{\Gamma}(w),
0 \big)\ \ \forall (t,x)\in D.
\end{equation}
Therefore in order to bound the function $w$ we only need to bound its 
values on the parabolic boundary. First, we estimate the left hand side of 
(\ref{bounds}).
For $a\leq x\leq b$ we have that $v^{(\ve)}(0,x)=\ve\eta(x)\leq
 \ve $, and so  $\beta^{(\ve)}(v^{(\ve)})\leq 0$. In view of (\ref{-P_t}),
 (\ref{f(x)AP detail}) and the definition of $f^{(\ve)}$ above there exists
  $\ve_0>0$ such that for every  $0<\ve \leq \ve_0$,
\begin{eqnarray*}
&w(0,x)\geq f(x)+\ve(\eta_{xx}(x)+\mu\eta_{x}(x)-r\eta(x))\\
&=-P_t(T',x)+\ve(\frac{\kappa^2}{2}\eta_{xx}(x)+\mu\eta_{x}(x)-r\eta(x))
\geq 0\,\,\,\ \ \forall\ a \leq x\leq b.
\end{eqnarray*}
On the interval $a_0\leq x \leq  a$ we have $\eta = 0$ and since 
$v^{(\ve)}(0,x)=\ve\eta(x)=0$ we see that $\beta^{(\ve)}
(v^{(\ve)}(0,x))=-1$. Since on this interval $f(x)\geq -rK$ we obtain
$$
w(0,x)=f^{(\ve)}(x)-rK\beta^{(\ve)}(0)\geq f(x)+rK\geq 0\ 
\,\mbox{for}\,\ a_0 \leq x< a.
$$
Hence, $w\geq 0$ on $\Gamma_2$.  We obtain next that,
\[
w(t,b)=-\frac{\partial P}{\partial t}(T'-t,b)\geq 0\,\,\mbox{on}\,\,\Gam_1\,\,
\mbox{and}\,\, w(t,b)=0\,\,\mbox{on}\,\,\Gam_3.
\]
It follows that $\min\big (\min_{\Gamma}(w),0 \big )=0.$

Next, we estimate the right hand side of (\ref{bounds}). On $\Gamma_2$ we 
have that
$$
w(0,x)\leq |f^{(\ve)}(x)+\ve(\frac{\kappa^2}{2}\eta_{xx}(x)+
\mu\eta_{x}(x)-r\eta(x))-rK\beta^{(\ve)}(v^{(\ve)}(0,x))  |
$$
$$
\leq \sup|f^{(\ve)}|+\ve\sup|\frac{\kappa^2}{2}\eta_{xx}+
\mu\eta_{x}-r\eta|+rK  \leq C_0
$$
where $C_0>0$ is a constant independent of $\ve$, and so
$$
0=\min\big (\min_{\Gamma}(w),0 \big)\leq w(t,x)\leq  \max\big(\max(-
\frac{\partial P}{\partial t}(T'-t,b),C_0)  \big)=C_1.
$$
We conclude that there are some constants $\ve_0$ and $C_1 $ such 
that for every $0<\ve\leq \ve_0$,
\begin{equation}\label{C_1}
0\leq v^{(\ve)}_{t}(t,x)\leq C_1.
\end{equation}
Since $v^{(\ve)}_{t}\geq 0$ and $v^{(\ve)}(0,x)\geq 0 $ for 
$a_0\leq x\leq b$ we deduce that $v^{(\ve)}(t,x) \geq 0$
and because $v^{(\ve)}_{t}$ is uniformly bounded it follows that 
$v^{(\ve)}$ is also uniformly bounded. By the properties of $\be^{(\ve)}$
we see that
\begin{equation}\label{beta.v}
-1\leq\beta^{(\ve)}(v^{(\ve)})\leq 0\ \ or\ \ \|\beta^{(\ve)}
(v^{(\ve)})\|_{L^{\infty}[D]}\leq1.
\end{equation}
Let $D_0=(0,T)\times(a_2,b)$ be an upper subrectangle of $D$ where $a_2$ is 
the same as in the definition of the function $\eta$ in (4). From the 
definition we have $v^{(\ve)}(0,x)=\ve\eta(x)=\ve$ in $\bar D_0$ and since 
 $v^{(\ve)}_t$ is nonnegative, we obtain that $v^{(\ve)}(t,x)
 \geq\ve$,  and so  $\beta^{(\ve)}(v^{(\ve)}(t,x))=0.$

This means that on $D_0$ the function $v^{(\ve)}$ satisfies the 
parabolic equation
$$
\textbf{L}v^{(\ve)}=f^{(\ve)}.
$$ 
For $w=v^{(\ve)}_{t}$ and $0<\ve<\ve_0$ we also have that  
$$
\textbf{L}w(t,x)=0\ 
\ \forall (t,x)\in D_0\,\,\mbox{and}\,\, w(0,x)=f^{(\ve)}(x)\,\,\mbox{when}\,\,
  a_2\leq x\leq b.
$$ 

Next, let $y(t,x)$ be a function on $D_0$ such that 
$$
\textbf{L}y(t,x)=0\ \ 
\forall (t,x)\in D_0,\ \ y(0,x)=f^{(\ve)}(x) \ \forall a_2\leq x\leq b
$$ 
and all of 
its first and second order derivatives are bounded there. Such a function
exists since we can choose a smooth function on the remaining part 
$\Gam_0\setminus\{0\}\times[a_2,b]$ of the parabolic boundary $\Gam_0$ of $D_0$
which extends $f^{(\ve)}(x)$ as a smooth function to the whole $D_0$,
and then 
use Theorem 12 from Chapter 3 in ~\cite{KSt}. For each $\ve<\ve_0$ 
we define $z(t,x)=w(t,x)-y(t,x)$ in the domain $D_0$ where $w(t,x)=
v^{(\ve)}_{t}(t,x)$. Then $z(0,x)=w(0,x)-y(0,x)=0$ for every 
$a_2\leq x\leq b$. Fix $x_0\in(a_2,b)$, then by Proposition 4.5 from 
Section 4.1
 of ~\cite{KSt} we obtain that $|z_t(t,x_0)|,|z_{tt}(t,x_0)|<C$ for every 
 $0\leq t\leq T'$ where a constant $C>0$ is independent of $\ve$.
Since we assume that  $|y_t(t,x_0)|,|y_{tt}(t,x_0)|<C_1$ for $0\leq t\leq T'$ 
it follows that $|w_t(t,x_0)|,|w_{tt}(t,x_0)|<C_1+C$  (for every $\ve$).
Let  $D_1= (x_0,b)\times(0,T)$, then by Theorem $6$ in chapter 3 of
 ~\cite{F} we obtain that for every $\ve>0$ (and in fact every 
 $0<\alpha<1$) $v^{(\ve)},v^{(\ve)}_{t} \in \bar C_{2+\alpha}[D_1]$
and  there is a constant $C$ independent of $\ve$ such that  
$$
\overline{|v^{(\ve)}|}_{2+\alpha}+\overline{|v^{(\ve)}_{t}|}_{2+
\alpha}<C.
$$
In particular, we get
\begin{equation}\label{4.12-old} 
\|v^{(\ve)}_{x}\|_{L^{\infty}(D_1)}+
\|v^{(\ve)}_{tx}\|_{L^{\infty}(D_1)}+\|v^{(\ve)}_{tt}\|_{L^{
\infty}(D_1)}<C.
\end{equation}

Considering again the whole region $D$ we have
$$
-\int_x^{b}v^{(\ve)}_{xx}(t,y)dy=\frac{2}{\kappa^2}\int_{x}^{b}\big(
\mu v^{(\ve)}_{x}(t,y)-v^{(\ve)}_{t}(t,y)-rv^{(\ve)}(t,y)+
\beta^{(\ve)}(v^{(\ve)}(t,y))+f^{(\ve)}(y)\big)dy.
$$
Hence,
\begin{eqnarray*}
&v^{(\ve)}_{x}(t,x)=v^{(\ve)}_{x}(t,b)+\frac{2}{\kappa^2}\big [
\mu\big(v^{(\ve)}(t,b)-v^{(\ve)}(t,x)\big)\\
&+\int_{x}^{b}\big( -v^{(\ve)}_{t}(t,y)-rv^{(\ve)}(t,y)+\beta^{(\ve)}
(v^{(\ve)}(t,y))+f^{(\ve)}(y) \big) dy\big ].
\end{eqnarray*}
Since all terms in the right hand side are uniformly bounded there is a 
constant $C>0$ independent of $\ve $ such that $\|v^{(\ve)}_{x}\|<C$
 for every $0<\ve\leq\ve_0$. Now we see that in the equation 
 $$
 v^{(\ve)}_{xx}(t,y)=\frac{2}{\kappa^2}\big(\mu 
 v^{(\ve)}_{x}(t,y)-v^{(\ve)}_{t}(t,y)-rv^{(\ve)}(t,y)+
 \beta^{(\ve)}(v^{(\ve)}(t,y))+f^{(\ve)}(y) 
 $$
all terms in the right hand side are uniformly bounded and therefore 
the term in the left is uniformly bounded, as well.

We summarize this in the following lemma.
\begin{lem}\label{lem3.6}
There are constants $C>0,\ve_0>0$  such that for every $\ve\leq
 \ve_0$,
$$
\|v^{(\ve)}_{xx} \|_{L^{\infty}[D]}+\|v^{(\ve)}_{x} 
\|_{L^{\infty}[D]}+\|v^{(\ve)}_{t} \|_{L^{\infty}[D]}+\|v^{(\ve)} 
\|_{L^{\infty}[D]}\leq C. 
$$
\end{lem}
We now obtain the following (see ~\cite{KSt}).
\begin{prop}
For any $1<p<\infty$  and $t\in [0,T']$, $v^{(\ve)}\to v$ as $\ve\to 0$ weakly
in $W^{1,p}(D)$. Furthermore, $v^{(\ve)}\to v$ uniformly on $D$ and also 
$v^{(\ve)}_{x}\to v_{x}$ uniformly in $x\in[0,K]$ for each  $t\in[0,T']$. The
function $v$ is the unique solution of v.i. Problem 1.
\end{prop}

Next, we analyze properties of second order derivatives starting
 with the following result.
\begin{lem}\label{sd_lem}
There is a constant $C>0$ such that for any $0<\ve \leq \ve_0$,
$$
\int_{0}^{T'} \int_{a_0}^{b}(v^{(\ve)}_{tx}(t,x))^2dxdt<C.
$$
\end{lem}
\begin{proof}
 Set $v=v^{(\ve)},\ \beta=\beta^{(\ve)}$ and $w=v^{(\ve)}_{t}$.
Multiply the equation (\ref{w=v_epst}) by $w$ to obtain
$$
ww_t-\frac{\kappa^2}{2}ww_{xx}-\mu ww_{x}+(r+rK\beta'(v))w^2=0. 
$$
Integrating this equation over $(a_0,b)$ and recalling that $\beta'(v),t$ and
 $K$ are non-negative we obtain that for any $0\leq t\leq T'$,
\begin{equation}\label{sd_lem_eq1}
\frac{1}{2}\frac{d}{dt}\int_{a_0}^{b}w^2(t,x)dx-\frac{\kappa^2}{2}
\int_{a_0}^{b}w(t,x)w_{xx}(t,x)dx-
\mu \int_{a_0}^{b}\frac{1}{2}\frac{dw^2}{dx}(t,x)dx\leq 0.
\end{equation}
By (\ref{Lv=F_eps}) and (\ref{C_1}) we estimate the third term in 
(\ref{sd_lem_eq1}),
$$
\big | \mu\int_{a_0}^{b}\frac{1}{2}\frac{dw^2}{dx}(t,x)dx\big| =|\mu|
|w^2(t,b)-w^2(t,a_0)|=|\mu|w^2(t,b)<C^2_1.
$$
 
For the second term in (\ref{sd_lem_eq1}) we see that
$$
-\frac{\kappa^2}{2}\int_{a_0}^{b}w(t,x)w_{xx}(t,x)dx=\frac{\kappa^2}{2} 
\big ( \int_{a_0}^{b}w^2_{x}(t,x)dx-w(t,b)w_{x}(t,b)+w(t,a_0)w_{x}(t,a_0)
\big).
$$
Since $w(t,a_0)=0$ and the function $w(t,x)$ is uniformly bounded in $D$ 
we see in view of (\ref{4.12-old})
that $w_{x}=v^{(\ve)}_{tx}$ is uniformly bounded near the boundary
  $[0,T']\times\{b\}$ and $w(t,a_0)w_{x}(t,a_0)=0$ while 
$|w(t,b)w_{x}(t,b)|<C_2$ for some constant $C_2>0$ independent of $\ve$.
Thus, we conclude from (\ref{sd_lem_eq1}) that
$$
\frac{1}{2}\frac{d}{dt}\int_{a_0}^{b}w^2(t,x)dx+\frac{\kappa^2}{2}
\int_{a_0}^{b}w_x^2(t,x)dx\leq C_3
 $$
for some $C_3>0$ independent of $\ve$. Integrating the last equation over
 $[0,T']$ we obtain
$$
\frac{\kappa^2}{2}\int_{0}^{T'} \int_{a_0}^{b}w^2_{x}(t,x)dxdt+\frac{1}{2}
\int_{a_0}^{b}(w^2(T,x)-w^2(0,x))dx\leq C_3.
$$
Since the function $w$ is uniformly bounded it follows that there is $C>0$
 independent of $\ve$ such that
$$
\int_{0}^{T'} \int_{a_0}^{b}w^2_{x}(t,x)dxdt\leq C.
$$
\end{proof}

We will now deal with the $L^2$ properties of the function 
$v^{(\ve)}_{tt}(t,x)$.
\begin{lem}\label{sd_lem2}
There is a constant $C>0$ such that for any $0<\ve\leq \ve_0$
 and every $0<\sigma\leq t \leq T'$,
$$
\int_{a_0}^{b}(v^{(\ve)}_{tx}(t,x))^2dx+\int_{\sigma}^{t}
\int_{a_0}^{b}(v^{(\ve)}_{tt}(x,t))^2dxd\tau \leq \frac{C}{\sigma}.
$$
\end{lem}
\begin{proof}
Set $v=v^{(\ve)},\beta=\beta^{(\ve)}$ and $w=v^{(\ve)}_{t}$.
Multiplying (\ref{w=v_epst}) by the function $w_t$ we have
$$
w_t^2-\frac{\kappa^2}{2}w_{xx}w_{t}-\mu w_{x}w_{t}+(r+rK\beta'(v))ww_t=0 
$$
and an integration with respect to $x$ over $(a_0,b)$ yields
\begin{equation}\label{sd_lem2_eq1}
\int_{a_0}^{b}w_t^2dx-\frac{\kappa^2}{2}\int_{a_0}^{b}w_{xx}w_{t}dx-\mu
\int_{a_0}^{b} w_{x}w_{t}dx+\int_{a_0}^{b}(r+rK\beta'(v))ww_tdx=0.
\end{equation}
Fix some  $t\in [0,T]$. Since $w_{t}(t,a_0)=w(t,a_0)=0$ we see that
$$
\frac{\kappa^2}{2}\int_{a_0}^{b}w_{xx}(t,x)w_{t}(t,x)dx=\frac{
\kappa^2}{2}w_{x}(t,b)w_{t}(t,b) -\frac{\kappa^2}{4}\frac{d}{dt}
\int_{a_0}^{b}w_{x}(t,x)^2dx.
$$
From (\ref{4.12-old}) it follows that $\frac{\kappa^2}{2}|w_{x}(t,b)w_{t}
(t,b)|\leq C_1$ for some constant $C_1>0$ independent of $\ve$, and so we
 obtain
\begin{eqnarray}\label{sd_lem2_eq2}
&\frac{\kappa^2}{4}\frac{d}{dt}\int_{a_0}^{b}w_{x}^2(t,x)dx+\int_{a_0}^{b}
w_t^2dx+\int_{a_0}^{b}(r+rK\beta'(v))w(t,x)w_t(t,x)dx\\
&\leq \mu\int_{a_0}^{b} w_{x}(t,x)w_{t}(t,x)dx +C_1.
\nonumber\end{eqnarray}

Now we deal with the last term in (\ref{sd_lem2_eq1}). Since $\beta''(v)\leq0$
 and $v,w\geq 0$ we obtain that
\begin{eqnarray*}
&\int_{a_0}^{b}(r+rK\beta'(v))ww_tdx=\frac{1}{2}\int_{a_0}^{b}(r+rK\beta'(v))
\frac{d}{dt}w^2dx\\
&=\frac{1}{2}\frac{d}{dt}\int_{a_0}^{b}(r+rK\beta'(v))w^2(t,x)dx-
\frac{1}{2}\int_{a_0}^{b}rK\beta''(v)w^3dx\geq\frac{1}{2}\frac{d}{dt}
\int_{a_0}^{b}(r+rK\beta'(v))w^2dx. 
\end{eqnarray*}
We plug this inequality into (\ref{sd_lem2_eq2}) and obtain
$$
\frac{1}{2}\frac{d}{dt}\int_{a_0}^{b}[\frac {\ka^2}{2}w_{x}^2(t,x)+
(r+rK\beta'(v))w^2(t,x)]dx
+\int_{a_0}^{b}w_t^2(t,x)dx+\leq \mu\int_{a_0}^{b} w_{x}(t,x)w_{t}(t,x)dx
 +C_1.
 $$
Integrate the last inequality with respect to $\tau'$ over the interval
$(\tau,t)$ to obtain
\begin{eqnarray*}
&\frac{1}{2}\int_{a_0}^{b}[\frac {\ka^2}{2}w^2_{x}(t,x)+(r+rK\beta'(v))
w^2(t,x)]dx+\int_{\tau}^t\int_{a_0}^{b}w_t^2(t,x)dxd\tau'\\
&\leq  \mu\int_{\tau}^t\int_{a_0}^{b} |w_{x}(t,x)||w_{t}(t,x)|dxd\tau'+ 
C_1 (t-\tau)+\frac{1}{2}\int_{a_0}^{b}[\frac {\ka^2}{2}w^2_{x}(\tau,x)+
(r+rK\beta'(v))w^2(\tau,x)]dx. 
\end{eqnarray*}
Next, integrating in $\tau$ over the interval $(0,\sigma)$ for some 
$0<\sigma<t$ and taking into account that $(r+rK\be'(v))w^2\geq 0$ by the 
property (2) of $\be$ we obtain that
\begin{eqnarray}\label{sd_lem2_eq3}
&\frac{\sigma}{2}\int_{a_0}^{b}\frac {\ka^2}{2}w^2_{x}(t,x)dx+ \int_{0}^{\sigma}
\int_{\tau}^t\int_{a_0}^{b}w_t^2(t,x)dxd\tau'd\tau\\
& \leq 
C_2 +|\mu|\int_{0}^{\sigma}\int_{\tau}^t\int_{a_0}^{b} |w_{x}(t,x)||w_{t}(t,x)|
dxd\tau'd\tau +\frac{1}{2}\int_{0}^{\sigma}\int_{a_0}^{b}[\frac {\ka^2}{2}
w^2_{x}(\tau,x)+(r+rK\beta'(v))w^2(\tau,x)]dxd\tau.\nonumber
\end{eqnarray}
Now, by (\ref{C_1}), (\ref{beta.v}) and Lemma \ref{sd_lem} together with
the Cauchy--Schwarz inequality we estimate the right hand side of
(\ref{sd_lem2_eq3}) by a constant $C_3>0$ independent of $\ve$. Hence,
\[
C_3\geq \frac {\sig\ka^2}{4}\int_{a_0}^b w^2_xdx +\int_0^\sig\int_\tau^t
\int_{a_0}^bw^2_tdxd\tau'd\tau\geq\frac {\sig\ka^2}{4}\int_{a_0}^bw^2_xdx
+\sig\int_\sig^t\int_{a_0}^bw_t^2dxd\tau'
\]
and Lemma \ref{sd_lem2} follows.
\end{proof}

As a corollary of previous results we obtain 
\begin{prop}\label{tem2.7}
Let $\beta<\sigma<T$ and  $a<s(0)<s(\sigma)<b<\ln K$. Define $D^{\sigma}=
(0,\sigma)\times(a,b)$. Then
\begin{equation}\label{ee_2.11}
P(t,x)\in H^2(D^{\sigma})
\end{equation}
where by definition $H^2[U]$ is the set of all the functions in $L^2[U]$ 
with an $L^2$ weak second order derivatives.
Also there exists $C>0$ such that for every $0\leq t\leq T'$,
\begin{equation}\label{2.12}
\int_{a }^{b}|\frac{\partial^2 P}{\partial x\partial t}(t,x)|^2dx=
 \|\frac{\partial^2 P}{\partial x\partial t}(t,x) \|_{L^2[a,b]}<C.
\end{equation}
\end{prop}
\begin{proof}
From Lemma \ref{sd_lem2}, Lemma \ref{sd_lem} and Lemma \ref{lem3.6} we obtain 
that $\{v^{(\ve)}\}_{\ve<\ve_0}$ are uniformly bounded in 
$H^{2}[D^{\sigma}]$ and so they have a weak limit $\tilde{v}\in 
H^{2}[D^{\sigma}]$. Since $v^{(\ve)}\to v$ uniformly we must have 
that $v=\tilde{v}$, and  so $v\in H^{2}[D^{\sigma}]$.
Since $v$ is the solution of (\ref{3.25+}) we can apply Proposition 
\ref{prop heat eq} and
using the fact that the constant $C$ in (\ref{prop eq res}) doesn't depend on
$t$ we can obtain in a similar way  that for a fixed $\sigma$ there is a 
constant $C>0$ such that for every $0\leq t\leq \sigma$,  
   $$
   \|v_{x,t}(t,\cdot ) \|_{L^2[a_0,b]}<C.
   $$
From (\ref{rem4.7}) we can deduce the same result for the function $P(t,x)$.
\end{proof}
\begin{cor}\label{cor11}
For each $0\leq t<T$ the function $v_{t}(t,x)$ is Holder continuous with a
 Holder exponent $\frac{1}{2}$.
\end{cor}
\begin{proof}
For every $0<t<T $ Proposition \ref{tem2.7} gives us that $v_{t}(t,x)\in 
H^1[a_0,b]$. Hence, the result is a consequence of the Sobolev inequality.
\end{proof}
\begin{cor}\label{cor12}
For every $0\leq t<T'$ the functions $P_t(t,x)$ and $P_{xx}(t,x)$ as functions
 of $x$ are continuous in the closed interval $[s(t),b]$.
\end{cor}
\begin{proof}
For the function $P_t(t,x)$ the result follows from (\ref{rem4.7}) and the
 previous corollary. Since $P(t,x)$ is a solution of (\ref{P,A})
 in the interval $\{(t,x):s(t)<x< \ln K\}$ and since
 the functions $P_x(t,x)$ and $P(t,x)$ are continuous in the interval
   $[s(t),b]$ we obtain the result for $P_{xx}$, as well.
\end{proof}
\begin{cor}\label{cor2.8}
Let $\beta<\sigma<T$ and  $a<s(0)<s(\sigma)<b<\ln K$. Define $E=
\{(t,x): 0<t<\beta, a-\mu t <x<b-\mu t \}$ and $u(t,x)=e^{-rt}P(t,x
+\mu t)$. Then
\begin{equation}\label{2.13}
u(t,x)\in L^2[E]
\end{equation}
and there exists $C>0$ such that for every $0\leq t\leq \beta$,
\begin{equation}\label{2.14}
\int_{a-\mu t }^{b-\mu t}|\frac{\partial^2 u}{\partial x\partial t}
(t,x)|^2dx<C.
\end{equation}
\end{cor}
\begin{proof}
The assertion (\ref{2.13}) follows from Proposition \ref{tem2.7} and the 
definition of $u(t,x)$. For (\ref{2.14}) note that
$$
\frac{\partial^2 u}{\partial x\partial t}(t,x)=e^{-rt}\big (-r 
\frac{\partial P}{\partial x}(t,x+\mu t)+\mu\frac{\partial^2 P}{\partial
 x^2}(t,x+\mu t)+\frac{\partial^2 P}{\partial x\partial t}(t,x+\mu t)
  \big),
  $$
then use (\ref{2.12}) and the fact that for $(t,x)\in E$ the 
functions $\frac{\partial^2 P}{\partial x^2}(t,x+\mu t)$ and $\frac{\partial 
P}{\partial x}(t,x+\mu t)$ are bounded.
\end{proof}

\section{Price function near the writer's exercise boundary}\label{sec4}

\subsection{Regularity properties of price function}
Let $F(t,x)$ be the price function of the put game option 
(see Section \ref{sec2}). We begin this section by showing that near 
the writer's exercise region $\Gamma_1=\{(t,K): 0\leq t \leq \beta\}$ the 
function $\frac{\partial F}{\partial t}$ is continuous. Let
\begin{equation}\label{Y_t}
Y^{[s,x]}_t=(Y^{1,[s,x]}_t,Y^{2,[s,x]}_t)=(s+t,S^x_t)
\end{equation}
which is a non homogeneous in time Markov process
in $\bbR^+\times \bbR$ where  $S^x_t=xe^{\mu t+\kappa
 B_t}$ and $\mu=r-\frac{\kappa^2}{2}$. Let
\begin{equation}\label{L_Y}
\textbf{L}_{Y}=\frac{\partial }{\partial t}+\frac{\kappa^2x^2}{2}\frac{
\partial^2 }{\partial x^2}+ rx\frac{\partial }{\partial x}-r
\end{equation}
which is the infinitesimal generator of $Y_t$ when considered
on the space of all $C^2$ functions. This is a parabolic operator
 with bounded smooth coefficients in the domain
\begin{equation}\label{domD}
D=(0,\beta)\times(k,K)
\end{equation}
where $k>0$. Let $\bfP_{[s,x]}$ and $\bfE_{[s,x]}$ be the probability and
the corresponding expectation for the Markov process $Y$ starting at the 
point $[s,x]$. We will first show that for any $t_0\in[0,\beta)$,
\begin{equation}\label{P[t,x][Y]}
\lim_{(t,x)\to (t_0,K)}\bfP_{[t,x]}[Y_{\tau} \in \Gamma_1]=1
\end{equation}
where $\tau=\tau(\Gam)$ and for any closed set $Q\subset\bbR_+\times\bbR$
we set $\tau(Q)$ to be the arrival time at the set $Q$ for a Markov process
under consideration which is $Y_t$ here.
Indeed, choosing an appropriate nonnegative function $\phi\leq 1$ on the 
boundary $\Gam$ and relying on Chapter 3 in \cite{F} we can choose $u(t,x)
\in C^{1,2}(D)$ which solves the equation $\textbf{L}_{Y}u=0$
 in $D$ and equals 1 on the boundary part $\Gamma_1$ for $0\leq t\leq t_1
 <\beta$ while decaying smoothly to $0$ when $t$ grows to $\beta$. Then 
$$
u(t,x)=\bfE_{[t,x]}\phi(Y_\tau)\leq \bfP_{[t,x]}\{Y_{\tau} \in \Gamma_1\},
$$
 and so
$$
1\geq \liminf_{(t,x)\to (t_0,K)}\bfP_{[t,x]}[Y_{\tau} \in \Gamma_1] \geq 
\lim_{(t,x)\to (t_0,K)}u(t,x)=u(t_0,K)=1.
$$

Next let $f(x)=(K-x)^+$ and $g(x)=f(x)+\delta$.
Recall that the price of a put game option with an expiration time $T$ and 
a constant penalty $\delta$ can be written in the form
$$
F(t,x)=\sup_{0\leq \tau \leq \tilde T}\inf_{0 \leq \sigma  \leq \tilde 
T}J_{[t,x]}(f,g,\sigma,\tau)
$$
where $ \tilde T=\inf\{t:Y^{1}_t=T\}$ and for any bounded Borel functions
$\hat f$ and $\hat g$ we write
$$
J_{[t,x]}(\hat f,\hat g,\sigma,\tau)=\bfE_{[t,x]}[e^{-r\sigma \wedge \tau}
 (\hat g(Y^2_{\sigma})\bbI_{\{\sigma <\tau\}}+\hat 
 f(Y^2_{\tau})\bbI_{\{\tau\leq \sigma\}})]. $$
Set
\begin{equation}\label{defF,f} 
f_s(x)=F(s,x)\,\,\mbox{when}\,\,\beta<s<T\,\,\mbox{and}\,\,
F_s(t,x)=sup_{0 \leq \tau \leq \tilde s}inf_{0 \leq \sigma \leq \tilde s}
J_{[t,x]}(f_s,g,\sigma,\tau).
\end{equation}
where $ \tilde s=\inf\{u:Y^{(1)}_u=s\}$. Let $<\sigma^*,\tau^*>$ and
$<\sigma^*_s,\tau^*_s>$ be the two saddle points (see ~\cite{Ki}) 
corresponding to the optimal stopping games with values $F(t,x)$ and 
$F_s(t,x)$, respectively, and so
\begin{eqnarray}\label{saddle}
&\sigma^*=\inf\{0\leq t\leq \tilde T:F(Y_t)=g(Y^2_t) \},\ \ \ \tau^*=\inf
\{0\leq t\leq \tilde T:F(Y_t)=f(Y^2_t)\}\\
&\sigma_s^*=\inf\{0\leq t\leq \tilde s:F_s(Y_t)=g(Y^2_t) \},\ \ \ 
\tau^*_s=\inf\{0\leq t\leq \tilde s:F_s(Y_t)=f_s(Y^2_t)\}.\nonumber
\end{eqnarray}
Then
$$
F(t,x)=J_{[t,x]}(f,g,\sigma^*,\tau^*)\ \ and\ \  F_s(t,x)=J^s_{[t,x]}
(f_s,g,\sigma^*_s,\tau^*_s).
$$
\begin{lem}\label{lem4.3}
For all $0\leq t\leq s<T$ and $x>0$, $F_s(t,x)=F(t,x)$. 
\end{lem}
\begin{proof} We have
\begin{eqnarray*}
&F_s(t,x)=J_{[t,x]}(f_s,g,\sigma^*_s,\tau^*_s)\leq J_{[t,x]}
(f_s,g,\sigma^*,\tau^*_s)\\
&=\bfE_{[t,x]}[e^{-r\sigma^*\wedge\tau^*_s}
\big(f_s(Y^2_{\tau^*_s})\bbI_{\{\tau^*_s
\leq \sigma^*\}}+F(Y_{\sigma^*})\bbI_{\{\sigma^*<\tau^*_s\}}\big)]\nonumber\\
&\leq \bfE_{[t,x]}[e^{-r\sigma^*\wedge\tau^*_s}\big(F(Y_{\tau_s^*})
\bbI_{\{\tau^*_s
\leq \sigma^*\}}+F(Y_{\sigma^*})\bbI_{\{\sigma^*<\tau^*_s\}}\big)]=
\bfE_{[t,x]}[e^{-r\sigma^*\wedge\tau^*_s}F(Y_{\sigma^*\wedge\tau_s^*})]\leq F(t,x).
\end{eqnarray*}
Indeed, the first inequality above follows by the saddle point property.
The second inequality holds true since $F$ is nonincreasing in the time 
variable, $\tau^*_s\leq\tilde s=s-t$ for $Y^{[t,x]}$ and
$Y^{1,[t,x]}(\tau_s^*)\leq s$. The third inequality  is satisfied
since the process $e^{-rY^\bbI_{\sigma^*\wedge u}}F(Y_{\sigma^*\wedge u})$
is a continuous supermartingale in $u$ with respect to $\bfP_{[t,x]}$ 
(see ~\cite{IK}).
For the other direction we have
\begin{eqnarray*}
&F(t,x)\leq \bfE_{[t,x]}[e^{-r \tilde s \wedge\sigma^*_s\wedge\tau^* }
F(Y_{\tilde s\wedge\sigma^*_s\wedge\tau^*})]
=\bfE_{[t,x]}[e^{-r\tilde s\wedge\sigma^*_s\wedge\tau^* }\big(f(Y^2_{\tau^*})
\bbI_{\tau^* 
\leq \tilde s\wedge\sigma^*_s }\\
&+F(Y_{\tilde s})\bbI_{\tilde s < \tau^*\wedge\sigma^*_s}+g(Y^2_{\sigma_s^*})
\bbI_{\sigma^*_s<\tilde s \wedge\tau^*  } \big)]
\leq \bfE_{[t,x]}[e^{-r\tilde s\wedge\sigma^*_s\wedge\tau^* }
\big(f_s(Y^2_{\tau^*\wedge
 s })\bbI_{\tau^* \wedge\tilde s\leq\sigma^*_s }\\
 &+g(Y^2_{\sigma_s^*})\bbI_{\sigma^*_s<\tilde s \wedge\tau^*  } \big)]
=J_{[t,x]}(f_s,g,\sigma^*_s,\tau^*\wedge \tilde s))\leq J_{[t,x]}(f_s,g,
\sigma^*_s,\tau^*_s)=F_s(t,x)
\end{eqnarray*}
where we use the submartingale property of $e^{-rY^\bbI_{\tau^*\wedge u}}
F(Y_{\tau^*\wedge u})$ in $u$.
\end{proof}

Now for any bounded Borel functions $\hat f$ and $\hat g$ set 
$$
I_s(t,x,\hat f,\hat g)=\sup_{0 \leq \tau \leq \tilde s}\inf_{0\leq \sigma 
\leq \tilde s}J_{[t,x]}(\hat f,\hat g,\sigma,\tau).
$$
From the time homogeneity of the process $Y^2_t=S_t$ we obtain that
\begin{equation}\label{5.4}
I_{s+h}(t+h,x,\hat f,g)=I_s(t,x,\hat f,\hat g).
\end{equation}
\begin{prop}\label{prop4.4}
There is a constant $C>0$ such that for any $(t,x)\in (0,\beta)\times 
(k,K)$,
  $$
  0\leq -\frac{\partial F}{\partial t}(t,x)\leq C\bfP_{[t,x]}
[\tau^*_s<\sigma^*_{s+h} ].
$$
\end{prop}
\begin{proof} The left hand side of the above inequality follows from
(iii) and (iv) of Proposition \ref{firstProp}. For the right hand side,
let $h>0$ be such that $\beta+h<T-h$ and $ t+h<\beta $ and let 
$\beta<s<T-h$. By (see \cite{L}) the price function of an American put option 
has a bounded derivative with respect to $t$ in $[0,s+h]\times \bbR$, i.e.  
$C=\sup_{(t,x)\in [0,s+h]\times \bbR_+ }|\frac {\partial F_A(t,x)}{\partial t}
|<\infty.$ This together with Proposition \ref{firstProp}(ii) yields
\begin{equation}\label{partialF<C}
\sup_{\beta<s<T,x\geq 0}|\frac {\partial F(s,x)}{\partial s}|\leq C.
\end{equation}
 Next, by Lemma \ref{lem4.3} and the saddle point property,
\begin{equation}\label{4.12}
F(t,x)=F_s(t,x)=J_{[t,x]}(f_s,g,\sig^*_s,\tau_s^*)\leq 
J_{[t,x]}(f_s,g,\sig^*_{s+h},\tau_s^*).
\end{equation}
By Lemma \ref{lem4.3}, (\ref{5.4}) and the saddle point property,
\begin{eqnarray}\label{4.13}
&F(t+h,x)=F_{s+h}(t+h,x)=I_{s+h}(t+h,x,f_{s+h},g)=I_s(t,x,f_{s+h},g)\\
&=J_{[t,x]}(f_{s+h},g,\sig^*_{s+h},\tau_{s+h}^*)\leq
J_{[t,x]}(f_{s+h},g,\sig^*_{s+h},\tau_{s}^*).\nonumber
\end{eqnarray}
Now, (\ref{defF,f}), (\ref{partialF<C}), (\ref{4.12}) and (\ref{4.13}) 
yields that
\begin{eqnarray*}
&0\leq \frac{1}{h}(F(t,x)-F(t+h,x))
\leq\frac{1}{h}\bfE_{[t,x]}[e^{-r\sigma^*_{s+h}\wedge
\tau^*_{s}}(f_s(Y^2_{\tau^*_s})
-f_{s+h}(Y^2_{\tau^*_s}))\bbI_{\{\tau^*_s\leq \sigma^*_{s+h}\}}]\\
&\leq\frac{1}{h}
E[e^{-r\sigma^*_{s+h}\wedge\tau^*_{s}}Ch\bbI_{\{\tau^*_s\leq \sigma^*_{s+h}\}}]
\leq C\bfP_{[t,x]}[\tau^*_s \leq \sigma^*_{s+h}].
\end{eqnarray*}
Passing to the limit as $h\to 0$ we obtain the result.
\end{proof}
\begin{cor}
For every $0\leq t_0 <\beta$, $\lim_{(t,x)\to (t_0,K)}\frac{\partial
 F}{\partial t}(t,x)=0$,
and so $\lim_{(t,x)\to (t_0,\ln K)}\frac{\partial P}{\partial t}(t,x)=0$.
\end{cor}
\begin{proof}
In view of Proposition \ref{prop4.4} we only have to show that for every 
$0\leq t_0<\beta$, 
$$
 \lim_{(t,x)\to (t_0,K)}\bfP_{[t,x]}[\tau^*_s 
\leq \sigma^*_{s+h}]=0.
$$
Let $D$ be as in (\ref{domD}), $\Gamma_2=\{(\beta,x):k\leq x\leq K\}$ 
and $\Gamma_3=\{(t,k):0\leq t\leq \beta \}$. It follows from the definition 
of $\tau^*_s$ and $\sigma^*_{s+h}$ that for every $(x,t) \in D$,
$$
\{\tau(\Gamma_1)<\tau(\Gamma_2\cap\Gamma_3)\}\subset \{\sigma^*_{s+h}<
 \tau^*_s \}\ \ with\ respect\ to \ \bfP_{[t,x]}.
 $$
From (\ref{P[t,x][Y]}) we obtain
 $$
\lim_{(t,x)\to (t_0,K)}\bfP_{[t,x]}[
\sigma^*_{s+h}<\tau^*_s ]=1
$$
and the result follows.
\end{proof}
%%%%%%%%%%%%%%%%%%%%%%%%%%%%%%%%%THE FUNCTIONS  v and w%%%%%%%%%%%%%%%%%%%%%%%%%%%%%%%%%%
Next, we deal with functions $P(t,x)=F(t,e^x)$, and so it is natural to
consider the domain $D_0=(0,\beta)\times(k,\ln K)$ for some positive $k<\ln K$ 
(which is, essentially, the same domain after the space coordinate change)
and let
\begin{equation}
c=P_{A, t}(\beta,\ln K)=\lim_{x\to \log K}P_{A,t}(\beta,x).
\end{equation}
%%see the operator definition put it before this defintion
Let $v(t,x)$ be a function solving the equation $(\frac{\partial}{\partial t}
+\textbf{A})v(t,x)=0$ with $\textbf{A}$ defined by (\ref{A}) and satisfying
 the boundary conditions
\begin{equation}\label{function v}
v(t,\ln K)=c\ , \ v(t,k)=P_t(t,k)\,\,\mbox{for}\,\, 0\leq t\leq \beta\,\,
\mbox{and}\,\, v(\beta,x)=P_t(\beta,x)\,\,\mbox{for}\,\, k<x<\ln K.
\end{equation}
Since these boundary conditions are continuous then (see ~\cite{F}) they
are satisfied by a unique solution in $C^{1,2}[D]$ of the above equation.
Let $w(t,x)$  be a function on $\bar D_0$ such that
\begin{equation}\label{P=w+v}
P_t(t,x)=w(t,x)+v(t,x)\ \ \  \forall(t,x)\in \bar D_0\setminus {(\beta,\ln K)}.
\end{equation}
Thus, $w(t,x)\in C^{1,2}[D']$ and it satisfies the same parabolic equation in
 $D_0$ as $\frac{\partial P}{\partial t}(t,x)$ and $v(t,x)$. Its boundary values
  are
\begin{equation}\label{function w}
w(t,\ln K)=-c \ ,\ w(t,k)=0\,\,\mbox{for}\,\, 0\leq t\leq \beta\,\,\mbox{and}
\,\, w(\beta,x)=0\,\,\mbox{for}\,\, k<x<\ln K.
\end{equation}
From the continuity of $v(t,x)$ on $\bar D_0$ we see that it is bounded there
 and since $\frac{\partial P}{\partial t}$ is also bounded there we obtain the 
 same result for the function  $w$ as for $v$. Hence,
\begin{equation}\label{w and v}
w(t,x),v(t,x)\in C^{1,2}[D_0]\cap L^{\infty}[D_0].
\end{equation}
%%%%%%%%%%%%%%%%%%%%%%%%%%%%%%%%%%%%%%%%%%%%%%%%%%%%%%%%%%%%%%%%%%%%%%%%%%%%%%
%   THE FUNCTION w
%%%%%%%%%%%%%%%%%%%%%%%%%%%%%%%%%%%%%%%%%%%%%%%%%%%%%%%%%%%%%%%%%%%%%%%%%%%%%%
\subsection{Integrability of $w_t(t,x)$ and $w_x(t,x)$}
Now we will analyze the function $w(t,x)$.
Let $Z^{[u,x]}_{t}=(u+t,X^x_t)$ be the diffusion process in the plane whose 
infinitesimal generator is equal to $\textbf{L}_1=\frac{\partial}{\partial t}
+\textbf{A}$ on the space of $C^2$ functions.
For each $\ve>0$ define $D_{\ve}=(0,\beta-\ve)\times(k+\ve,\ln K-\ve)$.
Let $\Gamma_{\ve}$ be the parabolic boundary of $D_{\ve}$.
For every $\ve >0$ which is sufficiently small
we can find a smooth function $\bar w(t,x)$ with compact support on the plane
 such that in $\bar D_{\ve}$ it is equal to $w(t,x)$.
By the Dynkin formula we obtain that for every $(u,x)\in D_{\ve}$,
\begin{equation}
\bfE_{[u,x]}[\bar w(Z_{\tau(\Gamma_{\ve})})]=\bar w(u,x)+ \bfE_{[u,x]}[
\int_0^{\tau(\Gamma_{\ve})}\textbf{L}_1\bar w(Z_s)ds]
\end{equation}
where $\tau(Q)$ denotes the arrival time to $Q$ by the process $Z_t^{[u,x]}$.
Note that since $w(t,x)=\bar w(t,x)$ for $(t,x)\in \bar D_{\ve}$ we 
can replace $\bar w$ by $w$ in the above formula and since $Z^{[u,x]}_s\in D_0$
for $s\leq\tau$ we obtain that $\textbf{L}_1\bar w(Z_s)=0$. It follows that
 for every $\ve>0$,
\begin{equation}
w(u,x)=\bfE_{[u,x]}[w(Z_{\tau(\Gamma_{\ve})})].
\end{equation}
Now fix $(u,x)\in D_0$ and a continuous path $\omega_0$. Let $\cE= 
\{Z^{[u,x]}_{\tau(\Gamma_{\frac{1}{n}})}(\omega_0)\}_{n_0<n}\subset 
\bar D_0$ where $n_0$ is such that $(u,x)\in D_{\frac{1}{n_0}}$. The 
sequence of times $\{\tau(\Gamma_{\frac{1}{n}})(\omega)\}_{n>n_0}$ 
is non decreasing with respect to $n$ and so it has a limit $\rho\leq T$.
 Let $\gamma$ be an accumulation point in $\cE$, i.e. $\lim_{k\to 
 \infty}Z^{[u,x]}_{\tau(\Gamma_{\frac{1}{n_k}})}(\omega_0)= \gamma$ 
 for some subsequence $n_k$. Define $d(y,\Gamma_0)=
 \inf\{|y-x|:x\in \Gamma_0 \}$ and note that this function is continuous 
 on $\bar D_0$ and it is $0$ if and only if $y\in \Gamma_0$. Since
   $d(Y^{[u,x]}_{\tau(\Gamma_{\frac{1}{n_k}})}(\omega_0),\Gamma_0)\leq 
 \frac{1}{n_k}$ for each $k$ we conclude that $\gamma \in \Gamma_0 $ and since 
 $\lim_{k\to \infty} \tau(\Gamma_{\frac{1}{n_k}})=\rho$ it follows
  that $Z^{[u,x]}_{\rho}(\omega_0)=\lim_{k\to \infty }
  Z^{[u,x]}_{\tau(\Gamma_{ \frac{1}{n_k}})}(\omega_0)=\gamma.$
Hence, $\tau(\Gamma_0)(\omega_0)=\rho$.
By the definition $w(t,x)$ is continuous except at the point 
$(\beta,\ln K)$ but because
$\bfP_{[u,x]}[Z_{\tau(\Gamma_0)}=(\beta,\ln K)]=0$ for every $(u,x)\in D_0$ we can 
ignore paths that reach the point $(\beta,\ln K)$, and so
\begin{equation}
\lim_{\ve \to 0}w(Z^{[u,x]}_{\tau(\Gamma_{\ve})})=w(Z^{[u,x]}_{
\tau(\Gamma_0)})\ \ \ \bfP_{[u,x]}\ a.s.
\end{equation}

\begin{cor}
For every $(t,x)\in D_0$,
$$
w(t,x)=\bfE_{[t,x]}[w(Z_{\tau(\Gamma_0)})]=-c\bfE_{[t,x]}[\bbI_{\{\tau(
\Gamma_{01})< \tau(\Gamma_{02} \cup \Gamma_{03})\}}]=-c\bfP_{[t,x]}
[\tau(\Gamma_{01})<\tau(\Gamma_{02}\cup\Gamma_{03})]
$$
where $\Gam_{01}=\{ (t,\ln K):\, 0\leq t\leq\beta\}$, $\Gamma_{02}=
\{(\beta,x):k\leq x\leq\ln K\}$ and $\Gamma_{03}=\{(t,k):0\leq t\leq \beta \}$.
\end{cor}
\begin{proof}
From (\ref{w and v}) we know that the function $w(t,x)$ is bounded and so 
we can use the Lebesgue bounded convergence theorem and from the boundary 
conditions on $w(t,x)$ it follows that
$$
w(t,x)=\lim_{\ve \to 0}\bfE_{[t,x]}[w(Z_{\tau(\Gamma_{\ve})})]=
\bfE_{[t,x]}[\lim_{\ve \to 0}w(Z_{\tau(\Gamma_{\ve})})]=\bfE_{[t,x]}
[w(Z_{\tau(\Gamma_0)})]
$$
which gives the first equality of the corollary while the second 
equality follows from (\ref{function w}) and the third equality is
obvious. \end{proof}
Let $(t,x),(t',x) \in D_0$ and assume that $t \leq t'$. Then it is not
difficult to understand that
$$
\bfP_{[t,x]}[\tau(\Gamma_{01})<\tau(\Gamma_{02} \cup \Gamma_{03})] \geq 
\bfP_{[t',x]}[\tau(\Gamma_{01})<\tau(\Gamma_{02} \cup \Gamma_{03})],
$$
and so $w(t,x)$ is nonincreasing in $t$ for every $x$ which implies that
\begin{equation}\label{w_t>0}
\frac{\partial w}{\partial t}(t,x)\geq 0\ \ \forall (t,x)\in D_0.
\end{equation}
It is also easy to see that for $0\leq t<T$ and $0\leq x\leq x'\leq \ln K$, 
$$
\bfP_{[t,x]}[\tau(\Gamma_{01})<\tau(\Gamma_{02}\cup\Gamma_{03})]\leq 
\bfP_{[t,x']}[\tau(\Gamma_{01})<\tau(\Gamma_{02}\cup\Gamma_{03})],
$$
and so
\begin{equation}\label{w_x>0}
\frac{\partial w}{\partial x}\leq 0\,\,\,\, \ \forall (t,x)\in D_0.
\end{equation}
%[proof of theorem 2.5 (2)]
\begin{lem}\label{lem-thm-2.5-[2]}
The functions $w_t$ and $w_x$ are in $L^1[D_0]$.
\end{lem}
\begin{proof}
We will use (\ref{w_t>0}) in order to prove the result for $w_t(t,x)$. 
The case of $w_x(t,x)$ can be proven similarly be using (\ref{w_x>0}).
Using (\ref{function w}), (\ref{w_t>0}) and the 
continuity of $w(0,x)$ on $\{0\}\times[k,\ln K]$ we obtain that
\begin{eqnarray*}
&\int_{D_0}|\frac{\partial w}{\partial t}|dtdx=\int_{k}^{\ln K}\int_{0}^{
\beta}\frac{\partial w}{\partial t}dt dx =\lim_{\ve \to 0} 
\int_{k}^{\ln K-\ve}\int_{0}^{\beta}\frac{\partial w}{\partial t}dt dx\\
&=\lim_{\ve \to 0}\int_{k}^{\ln K-\ve}(w(\beta,x)-w(0,x))dx
=-\lim_{\ve \to 0}\int_{k}^{\ln K-\ve}w(0,x)dx=-\int_k^{\ln K}
w(0,x)dx<\infty.
\end{eqnarray*}
Using (\ref{w_x>0}) in place of (\ref{w_t>0}) the proof of integrability
 of $w_x$ is similar.
\end{proof}
%%%%%%%%%%%%%%%%%%%%%%%%%%%%%%%%%%%%%%%%%%%%%%%%%%%%%%%%%%%%%%%%%%%%%%%%%
%    THE FUNCTION v
%%%%%%%%%%%%%%%%%%%%%%%%%%%%%%%%%%%%%%%%%%%%%%%%%%%%%%%%%%%%%%%%%%%%%%%%%%%
\subsection{Integrability of $v_t(t,x)$ and $v_x(t,x)$}
We continue this section by analyzing the function $v(t,x)$ solving the
equation $\bfL_1v=0$ with the boundary conditions given by (\ref{function v}).
Let $C^{1,2}[\bar D_0]$ be the set of all functions which have one derivative 
in $t$ and two derivatives in $x$ both uniformly continuous in $D_0$. 
\begin{lem}\label{lem4.7}
There exist a function $z(t,x)\in  C^{1,2}[\bar D_0]$ such that    
\begin{equation}
z(t,x)=v(t,x)\ \ \forall (t,x)\in \Gamma_0.
\end{equation}
\end{lem}
\begin{proof}
Recall that $P_{A,t}(T,x)=P_t(T,k)$ for $k\leq x<\ln K$ and note that the
functions $P_{A,t}(T,x),  P_{A,t}(t,x)$ and $ P_t(t,k)$ as function of 
$(t,x)$ belong to the space $C^{1,2}[\bar D_0].$
Set
$$
\tilde z(t,x)=\frac{\ln K-x}{\ln K-k}\big(P_t(t,k)+P_{A,t}(T,x)- 
P_{A,t}(T,k)\big)+\frac{x-k}{\ln K-k}P_{A,t}(t,x).
$$
Then $\tilde z(t,x) \in C^{1,2}[\bar D]$ since it is a linear combination
 of functions from this space. We also have
\begin{eqnarray*}
&\tilde z(t,k)=\frac{\ln K-k}{\ln K-k}\big(P_t(t,k)+P_{A,t}(T,k)-
P_{A,t}(T,k)\big)=P_t(t,x)\ \ \forall 0\leq t\leq \beta\\
&\tilde z(t,\ln K)=P_{A,t}(t,\ln K)\,\,\mbox{when}\,\, 0\leq t\leq \beta\,\,
\mbox{and for all}\,\,k\leq x\leq \ln K,\\
&\tilde z(T,x)=\frac{\ln K-x}{\ln K-k}\big(P_t(T,k)+P_{A,t}(T,x)-
 P_{A,t}(T,k)\big)+\frac{x-k}{\ln K-k} P_{A,t}(T,x)= P_{A,t}(T,x). 
\end{eqnarray*}
Thus, we obtain
\begin{equation}\label{4.24}
z(t,x)=\frac{\ln K-x}{\ln K-k}\tilde z(t,x)+\frac{x-k}{\ln K-k}\tilde 
z(T,x)\in C^{1,2}[\bar D].
\end{equation}
Since
\begin{equation*}
z(t,k)=\tilde z(t,k)=P_t(t,x),\,\,z(t,\ln K)=\tilde z(T,\ln K)=
P_{A,t}(T,\ln K)=c\,\,\mbox{and}\,\, z(T,x)=\tilde z(T,x)=P_{A,t}(T,x)
\end{equation*}
it follows that
\begin{equation}\label{4.25}
z(t,x)=v(t,x)\ \ \forall (t,x) \in \Gamma.
\end{equation}
\end{proof}

Next, define $f(t,x)=-\textbf{L}z(t,x)$. From Lemma \ref{lem4.7} we obtain
 that $f(x,t)$ is bounded in $D_0$ and so it belongs to $L^p[D_0]$ for every 
 $1\leq p\leq \infty$. Set $\tilde v(t,x)=v(t,x)-z(x,t)$ and observe that 
$$
\textbf{L}\tilde v(t,x)=f(t,x)\ \ and \ \ \tilde v(t,x)=0\ \ \forall 
(t,x)\in \Gamma_0.
$$
We conclude that the function $\tilde v(t,x)$ is the unique solution of the
 following problem (see ~\cite{BF}).
\begin{thm}\label{thm4.7}
Let $1\leq p<\infty $ then for any  $f(t,x)\in L^p[D_0]$ there exists a unique
function $\tilde v$ such that

 (i)  $\tilde v \in L^p[0,T;W^{2,p}(0,1)]\cap L^p[0,T;W^{1,p}_0(0,1)]$,
 
  (ii) $\frac{\partial \tilde v}{\partial t}\in L^p[D_0]$,
  
  (iii) $\textbf{L}\tilde v(t,x)=f(t,x)$ for every $(t,x)\in D_0$,
  
 (iv) $\tilde v|_{\Gamma_0}=0$.
\end{thm}

From assertions (i) and (ii) of Theorem \ref{thm4.7} we obtain that the 
functions $\tilde v_x(t,x)$ and  $\tilde v_t(t,x)$ are both in $L^p[D]$ for
every $0\leq p<\infty$ and since $z(t,x)\in C^{1,2}[\bar D_0]$ we obtain
 the following.
\begin{cor}\label{cor_them2.5(3)}
For every $1\leq p<\infty$ the functions $v_{t}(t,x)$ and $v_{x}(t,x)$ belong 
to the space $L^p[D_0]$.
\end{cor}
We can now summarize most of the results of this section as follows.
\begin{prop}\label{thm2.5}
Let $s(\beta)<k<\ln K<k'$ and define 
$$
D_0=(0,\beta)\times(k,\ln K)\,\,\mbox{and}\,\, D_0'=(0,\beta)\times(\ln K,k').
$$ 
Then the function $P_t(t,x)$ is continuous at every point in the domain 
$\bar D_0\setminus \{(\beta,\ln K)\}$, and there exist two functions 
$w(t,x)$ and $v(t,x)$ on $D_0$ such that
\begin{equation}\label{P_t=w+v}
P_t(t,x)=w(t,x)+v(t,x)\ \ for\ every\ \ (t,x)\in \bar D_0\setminus 
\{(\beta,\ln K)\},
\end{equation}
\begin{equation}\label{them2.5(1)}
w(t,x),v(t,x)\in C^{1,2}(D_0)\cup L^{\infty}[D_0]
\end{equation}
and both functions are solutions of the parabolic equation $\textbf{L}_1u=0$.
Furthermore, $w(t,x)$ is continuous in $D_0$ and it satisfies
\begin{eqnarray}\label{w in them}
&w(t,\ln K)= P_{A,t}(\beta,\ln K)\,\,\mbox{and}\,\, w(t,b)=0\,\,\mbox{when}
\,\, 0\leq t\leq \beta,\\
& w(\beta,x)=0\,\,\mbox{when}\,\, k<x<\ln K\nonumber
\end{eqnarray}
and
\begin{equation}\label{w_t w_x int}
w_t(t,x),w_x(t,x)\in L^1[D].
\end{equation}
Finally, $v(t,x)\in C(\bar D)$ and for every $ 1\leq p<\infty$,
\begin{equation}\label{v_x v_t integr}
v_t(t,x),v_x(t,x)\in L^p[D].
\end{equation}
The same decomposition of $P_{t}(t,x)$ with the same properties holds true in
 the domain $D_0'$.
\end{prop}
\begin{proof}
Taking the same functions $v$ and $w$ as in (\ref{P=w+v})
we see that (\ref{them2.5(1)}) is actually the same as (\ref{w and v}) and 
the fact that both $v$ and $w$ are solution of $\textbf{L}_1u=0$ is clear
from their definitions. Next we see that (\ref{w in them}) is the same as 
(\ref{function w}), that (\ref{w_t w_x int}) is the same as Lemma 
\ref{lem-thm-2.5-[2]} and that (\ref{v_x v_t integr}) is, in fact, Corollary 
\ref{cor_them2.5(3)}.
Observe that we did not use in this section the fact that $k<\ln K$ so 
all the proofs are also applicable to the case $k'>\ln K$ and the domain $D_0'$.
\end{proof}

From (\ref{P_t=w+v}), (\ref{w_t w_x int}), (\ref{v_x v_t integr}) and estimating
$P_{xx}$ via other derivatives in view of the equation (\ref{P,A})
we obtain the following.
\begin{cor}\label{cor2.6}
Let  $\tilde D=\{(t,x): 0<t<\beta, \ \ k-\mu t <x<\ln K-\mu t\}$ and
\begin{equation}\label{u_defintion}
u(t,x)=e^{-rt}P(t,x+\mu t).
\end{equation}
Then
\begin{equation}\label{2.8}
\frac{\partial^2 u }{\partial t^2}\in L^1[\tilde D].
\end{equation}
\end{cor}

\subsection{Price function when initial stock price is large}
Let $F(t,x),P(t,x)$ and $u(t,x)$ be as above. Recall that in the domain $(0,T)
\times (\ln K,\infty)$ the function $P(t,x)$ satisfies the equation
$\bfL_1P=0$, it is continuous in the closure of $[0,T]\times[\ln K,\infty)$
and $P(T,x)=(K-e^x)^+=0$ for $x>\ln K$. Define
\begin{equation}
v(t,x)=u(T-t,\frac{\kappa}{\sqrt{2}}x+\ln K+|\mu|T)
\end{equation}
where $u$ is given by (\ref{u_defintion}) and set $G=(0,T)\times(0,\infty)$.
It follows from  Proposition \ref{thm2.5} that
\begin{enumerate}
\item $v(t,x)\in C^{1,2}[G]\cup C[\bar G]$,
\item $v_{xx}(t,x)=v_t(t,x)$ for every $(t,x)\in G$,
\item $v(t,0)=u(T-t,\ln K+|\mu|T)$ is continuous,
\item $v(0,x)=0$ for every $x>0$,
\item $v(t,x)$ is bounded (since $P(t,x)$ is).
\end{enumerate}
Since a bounded solution of the heat equation in $G$ is unique (see ~\cite{C})
then for every $(t,x)\in G$,
\begin{equation}
v(t,x)=-2\int_{0}^{t}\frac{\partial K}{\partial x }(t-\tau,x)v(\tau,0)d\tau
\,\,\mbox{where}\,\, K(t,x)=\frac{1}{\sqrt{4\pi t}}e^{-\frac{x^2}{4t}},
\end{equation}
and so 
$$
v(t,x)=\frac{1}{\sqrt{4\pi}}\int_0^{t}\frac{xe^{-\frac{x^2}{4(t-
\tau)}}v(\tau,0)d\tau}{(t-\tau)^{3/2}}.
$$
Differentiating $v$ we obtain polynomials $Q_{k,n}(s,x)$ such that for all
$k,n\in\bbN$,
$$
\frac{\partial^{k+n} v}{\partial^nt\partial^k x}(t,x)=\int_0^{t}Q_{n,k}((t
-\tau)^{-1/2},x)e^{\frac{-x^2}{4(t-\tau)}}v(\tau,0)d\tau.
$$
If $N$ is large enough and $c>0$ then
$\frac{(t-\tau)^N}{x^N}Q_{k,n}((t-\tau)^{-1/2},x)$
is a polynomial in $(t-\tau)^{1/2}$ and $1/x$ and it is bounded on 
$ (0,T)\times (c,\infty)$. Since $sup_{y\geq 0}y^{N}e^{-y}<\infty$
for any $N\in\bbN$ we can set $y=\frac{x^2}{4(t-\tau)}$ deriving that
for any $N\in\bbN$ and $(t,x)\in (0,T)\times(c,\infty)$,
\begin{eqnarray*}
&\frac{\partial^{k+n} v}{\partial^nt\partial^k x}(t,x)=\int_0^{t}
\frac{4^N(t-\tau)^N}{x^{2N}}Q_{n,k}((t-\tau)^{-1/2},x)y^Ne^{-y}v(\tau,0)d\tau\\
&\leq(\frac{4}{x})^{N}\int_0^{t}\big(\frac{(t-\tau)^N}{x^{N}}
Q_{n,k}(x,(t-\tau)^{-1/2})\big) y^Ne^{-y}v(\tau,0)d\tau \leq\frac{C}{x^N}
\end{eqnarray*}
For some $C=C(N)>0$. Hence, the following results hold true.
\begin{cor}
For any $k,n$ positive integers $k,n$  and $c>0$,
$$
\frac{\partial^{k+n}
v(t,x)}{\partial^k t\partial^n x}\in L^2[(0,T)\times(c,\infty)].
$$
\end{cor}
\begin{cor}\label{cor_infty}
Let $\frac{\sqrt{2}}{\kappa}(\ln K+|\mu|T)<k'$ and $\tilde G=\{(t,x):
\,0<t<\beta,\ k'-\mu <x<\infty \}$. Then
\[ 
\frac{\partial^2 u}{\partial t^2}(t,x)\in L^2[\tilde G].
\]
\end{cor}

\section{Proof of main theorem}\label{sec5}

We split the proof into two cases for $x\leq\ln K$ and for $x >\ln K$.
\subsection{Case $x \leq \ln K$}
We begin by proving the upper bound in (\ref{P1_err}). Since the option
holder can exercise at time 0 it is clear from the definition of $P(t,x)$
in (\ref{P(x,t)}) that $P(t,x)\geq \psi(x)$ for every $x>0$. Furthermore,
by Proposition \ref{firstProp}(iv) for each fixed $t$ the function $P(t,x)$
as a function of $x$ is nonincreasing.
Therefore, $P(t,x)\geq P(t,\ln K)=\delta $ when $x\leq \ln K$.
From the definition (\ref{sigma^n}) of the stopping time $\sigma^{(n)}$
it is not difficult to see that in the present case when $\sigma^{(n)}<T$,
$$
x+\mu \sigma^{(n)}+\kappa W_{\sigma^{(n)}}<\ln K,
$$ 
and so
\begin{equation}\label{3.4}
P(\sigma^{(n)},x+\mu \sigma^{(n)}+\kappa W_{\sigma^{(n)}})\geq \delta.
\end{equation}
Hence for every $\tau \in\cT^{(n)}$  we obtain,
\begin{eqnarray}\label{3.5}
&\bfE [e^{-r\tau\wedge\sigma^{(n)}}\big(\psi(x+\mu t+\kappa W^{(n)}_{\tau})
\bbI_{\tau\leq\sigma^{(n)}}+\delta\bbI_{\sigma^{(n)}<\tau} \big)]\\
&\leq\bfE [e^{-r\tau\wedge \sigma^{(n)}}\big(P(\tau\wedge\sigma^{(n)},x+\mu 
\tau\wedge\sigma^{(n)}+ \kappa W_{\tau\wedge\sigma^{(n)}}^{(n)})]=
\bfE [u(\tau\wedge\sigma^{(n)},X^{(n)}_{\tau\wedge\sigma^{(n)}})]).\nonumber 
\end{eqnarray}
By Proposition \ref{martingel prop 2.1},
\begin{equation}\label{3.6}
\bfE [u(\tau\wedge\sigma^{(n)},X^{(n)}_{\tau\wedge\sigma^{(n)}})])=u(0,x)+
\bfE [\sum_{j=1}^{h^{-1}(\tau\wedge\sigma^{(n)})}\cD u((j-1)h,X_{(j-1)h}^{(n)})]
\end{equation}
where, as before, $u(t,x) = e^{-rt}P(t,x+\mu t)$.
Taking the sup with respect to all $\tau \in \cT^{(n)}$ in the inequality 
(\ref{3.5}) and using the fact that $u(0,x)=P(0,x)$ we obtain that
\begin{equation}\label{3.7}
P^{(n)}_1(x)-P(0,x)\leq \sup_{\tau \in\cT^{(n)}}\bfE [\sum_{j=1}^
{h^{-1}(\tau\wedge\sigma^{(n)})}\cD u((j-1)h,X_{(j-1)h}^{(n)})].
\end{equation}
Thus, in order to bound $P^{(n)}_1(x)-P(0,x)$ from the above it suffices 
to find an upper bound of the right hand in (\ref{3.7}).

Next, we split the domain  $[0,T]\times \bbR$ into three parts 
\begin{eqnarray}\label{C_domain}
&\textbf{C}=\{(t,x)\in [0,T-h]:\mu t +x> s(t+h)+|\mu|h+\kappa\sqrt{h}\},\\
&\textbf{S}=\{(t,x)\in [0,T-h]:\mu t +x\leq s(t)-|\mu|h-\kappa\sqrt{h}\} 
\,\,\mbox{and}\nonumber\\
&\textbf{B}=\{(t,x)\in[0,T-h]\times\bbR:s(t)-|\mu|h-\kappa\sqrt{h}\leq \mu 
t+x\leq s(t+h) +|\mu|h+\kappa\sqrt{h}\}.\nonumber
\end{eqnarray}
In order to estimate the right hand side of (\ref{3.7}) we split it into 
three parts according to the domains   $\textbf{C}$, $\textbf{S}$ and 
$\textbf{B}$, i.e.
\begin{eqnarray}\label{3.8}
&\bfE [\sum_{j=1}^{h^{-1}(\tau\wedge\sig^{(n)})}\cD u((j-1)h,X_{(j-1)h}^{(n)})]=
\bfE [\sum_{j=1}^{h^{-1}(\tau\wedge\sig^{(n)})}\cD u((j-1)h,\\
&X_{(j-1)h}^{(n)})\bbI_{((j-1)h,X_{(j-1)h})\in\textbf{C}}]+\bfE 
[\sum_{j=1}^{h^{-1}(\tau\wedge\sig^{(n)})}\cD u((j-1)h,X_{(j-1)h}^{(n)})
\bbI_{((j-1)h,X_{(j-1)h})\in\textbf{S}}]\nonumber\\
&+\bfE [\sum_{j=1}^{h^{-1}(\tau\wedge\sig^{(n)})}\cD u((j-1)h,X_{(j-1)h}^{(n)})
\bbI_{((j-1)h,X_{(j-1)h})\in\textbf{B}}].\nonumber
\end{eqnarray}
By Proposition \ref{firstProp}(ii) after the time $\beta$ the prices of the 
American and game put options coincide which enables us to conclude that
$u(t,x)=e^{-rt}P_A(t,x+\mu t)$ for $t\geq\beta$ and that
the sets $\textbf{C}_{t\geq\beta}=\{(t,x)\in
\textbf{C}:\, t\geq\beta\}$, $\textbf{S}_{t\geq\beta}=\{(t,x)\in
\textbf{S}:\, t\geq\beta\}$ and $\textbf{B}_{t\geq\beta}=\{(t,x)\in
\textbf{B}:\, t\geq\beta\}$ are the same as the corresponding parts
of the domains $\bar C$, $\bar S$ and $\bar B$ introduced in ~\cite{L} for
the case of American put options.  
Therefore, we can use the following results from Sections 4.2 and 4.3
in ~\cite{L}.
\begin{prop}\label{prop 3.1}
There exists a constant $C>0$ such that for every $\tau\in\cT^{(n)}$,
\begin{equation}\label{3.9}
\bfE [\sum_{j=k_{\beta}}^{(\tau/h)\vee k_{\beta}}\cD |u((j-1)h,X_{(j-1)h}^{(n)})
|\bbI_{((j-1)h,X_{(j-1)h})\in\textbf{C}}]\leq C\big(\frac{\sqrt{\ln n}}{n}
\big)^{4/5},
\end{equation}
where $k_\beta=\min\{ k:\, kh\geq\beta\}$,  and
\begin{equation}\label{3.10}
\bfE [\sum_{j=k_{\beta}}^{(\tau/h)\vee k_{\beta}}\cD u((j-1)h,X_{(j-1)h}^{(n)})
\bbI_{((j-1)h,X_{(j-1)h})\in\textbf{B}}]\leq \frac{C}{n^{3/4}}.
\end{equation}
\end{prop}
Observe also that $P(t,x)=K-e^x$ in the domain $\textbf{S}$, and so we can 
use there Lemma 2 from Section 4 of ~\cite{L}.
\begin{lem}\label{lem 3.2}
For every $(t,x)\in\textbf{S}$ we have $\cD u(t,x)\leq 0$, and so
$$
\bfE [\sum_{j=1}^{h^{-1}(\tau\wedge\sig^{(n)})}\cD u((j-1)h,X_{(j-1)h}^{(n)})
\bbI_{((j-1)h,X_{(j-1)h})\in\textbf{S}}]\leq 0.
$$
\end{lem}
Thus, for an upper bound of the right side of (\ref{3.7}) we can ignore the
 second term in the right hand side of (\ref{3.8}) and estimate only two
 remaining terms starting with the first term in the right hand side of
   (\ref{3.8}).
\begin{prop}\label{prop 3.3}
There is a constant $C>0$ such that for all $n\in\bbN$,
\begin{equation}\label{3.11}
\bfE [\sum_{j=1}^{h^{-1}(\tau\wedge\sig^{(n)})}|\cD u((j-1)h,X_{(j-1)h}^{(n)})|
\bbI_{((j-1)h,X_{(j-1)h})\in\textbf{C}}]\leq Cn^{-3/4}.
 \end{equation}
\end{prop}
\begin{proof}
We have
\begin{eqnarray}\label{3.12}
&\bfE [\sum_{j=1}^{h^{-1}(\tau\wedge\sig^{(n)})}|\cD u((j-1)h,X_{(j-1)h}^{(n)})
|\bbI_{((j-1)h,X_{(j-1)h})\in\textbf{C}}]\\
&=\bfE [\sum_{j=1}^{(h^{-1}(\tau\wedge\sig^{(n)}))\wedge k_{\beta} }|\cD 
u((j-1)h,X_{(j-1)h}^{(n)})|\bbI_{((j-1)h,X_{(j-1)h})\in\textbf{C}}]  
\nonumber\\
&+\bfE [\sum_{j=k_{\beta}}^{(h^{-1}(\tau\wedge\sig^{(n)}))\vee k_{\beta}}
|\cD u((j-1)h,X_{(j-1)h}^{(n)})|\bbI_{((j-1)h,X_{(j-1)h})\in\textbf{C}}].
\nonumber\end{eqnarray}
Proposition \ref{prop 3.1} provides a bound for the second term in the right
hand side of (\ref{3.12}), and so it remains to deal only with the first term
there. Note that if $jh<\sigma^{(n)}\wedge\beta^{(n)}$ and $(jh,X^{(n)}_{jh})
\in \textbf{C}$ then
$$
\tilde c_1(j)=s(jh)-\mu jh+\kappa\sqrt h\leq X^{(n)}_{jh}  \leq \ln K-2
\kappa\sqrt{h}-\mu jh=\tilde c_2(j)
$$
where the equalities above are just definitions of $\tilde c_1$ and 
$\tilde c_2$.
Observe also that since $x<\ln K$ and  $jh<\sigma^{(n)}$ then by the 
definition of the stopping times $\sigma^{(n)}$ the process $X^{(n)}_{jh}
+\mu jh$ does not exceed $\ln K-2K\sqrt{h}$.
By Proposition \ref{dis_operator_prop 2.2},
\begin{equation}\label{3.13}
\cD u(t,x)=\frac{1}{\kappa}\int_{0}^{\sqrt{h}}dy\int_{-\kappa y}^{\kappa y}
dz\big (z\frac{\partial^2 u }{\partial t\partial x}(t+y^2,x+z)+\delta(u)
(t+y^2,x+z)  \big).
\end{equation}
Relying on the same computation as in Section 4 of ~\cite{L} we see that for 
$(t,x)\in \textbf{C}$ and $x< \ln K-|\mu|h-\kappa\sqrt{h}$,
\begin{equation}\label{3.14}
|\cD u(t,x)|\leq \frac{\sqrt{h}}{\kappa}\int_{t}^{t+h}ds \int_{x-\kappa 
\sqrt{h}}^{x+\kappa \sqrt{h}}dz|\frac{\partial^2 u }{\partial t^2}(s,z)|.
\end{equation}
Thus, for $0\leq j < k_{\beta}$,
\begin{eqnarray*}
&E(|\cD u(jh,X^{(n)}_{jh})|\bbI_{\{(jh,X^{(n)}_{jh})\in\bfC\}\cap\{jh< 
\sigma^{(n)}\}})
\leq \int_{\tilde c_1(j)}^{\tilde c_2(j)}|Du(jh,y)|d\bfP_{X_{jh}}(y)\\
&\leq \int_{\tilde c_1(j)}^{\tilde c_2(j)}\big ( \frac{\sqrt{h}}{2\kappa}
\int_{jh}^{jh+h}ds\int_{y-\kappa\sqrt{h}}^{y+\kappa\sqrt{h}}|\frac{\partial^2
 u }{\partial t^2}(s,z)|dz\big )d\bfP_{X_{jh}}(y)\\
&=\frac{\sqrt{h}}{2\kappa}\int_{jh}^{(j+1)h}ds\int_{\tilde c_1(j)+
\ka\sqrt h }^{\tilde c_2(j)+\ka\sqrt h}dz|\frac{\partial^2 u }{\partial t^2}(s,z)|\int_{max(\tilde 
c_1(j) ,z-\kappa\sqrt{h})}^{min(\tilde c_2(j),z+\kappa \sqrt{h})}
d\bfP_{X_{jh}}(y)\\
&\leq \frac{\sqrt{h}}{2\kappa}\int_{jh}^{(j+1)h}ds\int_{\tilde c_1(j)-
\kappa\sqrt{h}}^{\tilde c_2(j)+\kappa\sqrt{h}}dz|\frac{\partial^2 u }{
\partial t^2}(s,z)|\bfP [|X_{jh}^{(n)}-z|\leq \kappa \sqrt{h}].
\end{eqnarray*}
From (\ref{Berry-Esseen estimate}) we see that there is a constant $C>0$ 
independent of $j$ and $n$ such that
$$
\bfP [|X_{jh}^{(n)}-z|\leq \kappa\sqrt{h}]\leq \frac{C}{\sqrt{j+1}}.
$$
Hence, for $jh<\sigma^{(n)}$, 
\begin{eqnarray*}
&\bfE\big(|\cD u(jh,X^{(n)}_{jh})|\bbI_{(jh,X^{(n)}_{jh})\in\bfC}
\bbI_{jh<\sig^{(n)}}\big)\leq \frac{
\sqrt{h}}{2\kappa}\int_{jh}^{(j+1)h}ds\int_{\tilde c_1(j)-\kappa\sqrt{h}}^{
\tilde c_2(j)+\kappa \sqrt{h}}dz|\frac{\partial^2 u }{\partial t^2}(s,z)|
\frac{C}{\sqrt{j+1}}\\
&=\frac{Ch}{2\kappa}\int_{jh}^{(j+1)h}\frac{ds}{\sqrt{h(j+1)}}\int_{
\tilde c_1(j)-\kappa \sqrt{h}}^{\tilde c_2(j)+\kappa \sqrt{h}}dz|\frac{
\partial^2 u }{\partial t^2}(s,z)|\\
&\leq \frac{C_1}{n}\int_{jh}^{(j+1)h}\frac{ds}{\sqrt{s}}\int_{
\tilde c_1(j) -\kappa \sqrt{h}}^{\tilde c_2(j)+\kappa \sqrt{h}}dz|
\frac{\partial^2 u }{\partial t^2}(s,z)|. 
\end{eqnarray*}
Define
$$
c_1(t)=s(t)-\mu t,\ \  c_2(t)= \ln K-\mu t -\kappa\sqrt{h}
$$
where $s(t)=\ln (b(t))$ is the free boundary of the option holder 
and $b(t)$ was introduced at the beginning of Section \ref{sec3}.
Observe that for every $j$ and any $jh\leq s\leq(j+1)h$,
$$
\tilde c_1(j)-\kappa\sqrt{h}\geq c_1(s),\ \ \tilde c_2(j)+\kappa\sqrt{h}
\leq c_2(s).
$$
Summing up the above estimates we obtain 
\begin{eqnarray}\label{3.15}
&\sum_{j=0}^{k_{\beta}-1}\bfE(|\cD u(jh,X^{(n)}_{jh})|\bbI_{\{(jh,
X^{(n)}_{jh})\in \bfC\}\cap\{jh < \sigma^{(n)}\}})\\
&\leq\frac{C_2}{n}+\frac{C_1}{n}\int_{h}^{\beta}\frac{ds}{\sqrt{s}}\int_{ 
c_1(s)}^{c_2(s) }dz|\frac{\partial^2 u }{\partial t^2}(s,z)|\nonumber\\
&=\frac{C_2}{n}+\frac{C_1}{n}\Big(\int_{h}^{\sqrt{h}}\frac{ds}{\sqrt{s}}
\int_{c_1(s)}^{c_2(s)}dz|\frac{\partial^2 u }{\partial t^2}(s,z)|+\int_{
\sqrt{h}}^{\beta}\frac{ds}{\sqrt{s}}\int_{c_1(s)}^{c_2(s)}dz|
\frac{\partial^2 u }{\partial t^2}(s,z)|\Big)\nonumber
\end{eqnarray}
where the term $\frac{C_2}{n}$ comes from the first term 
$\bfE |\cD u(0,x)|$ of the sum 
 which can be estimated easily using the fact that $u_t(t,x)$ and 
 $u_{xx}(t,x)$ are bounded for small $t$.
 
Let $G=\{(t,x):0<t<\beta, c_1(t) <x< \ln K-\mu t \}$
and note that $G\subset E\cup\tilde D $ where $E$ and $\tilde D$ 
are defined in Corollaries \ref{cor2.8} and \ref{cor2.6} which imply that 
$\frac{\partial^2 u }{\partial t^2}(s,z)\in L^1[F]$. Hence,
\begin{equation}\label{3.16}
\int_{\sqrt{h}}^{\beta}\frac{ds}{\sqrt{s}}\int_{c_1(s)}^{c_2(s) }dz|
\frac{\partial^2 u }{\partial t^2}(s,z)|\leq C_1n^{1/4}\int_{
\sqrt{h}}^{\beta}ds\int_{c_1(s)}^{c_2(s) }dz|\frac{\partial^2 u }{
\partial t^2}(s,z)|\leq Cn^{1/4}.
\end{equation}
Next, we estimate the first integral in brackets in the right hand
 side of (\ref{3.15}). Let $s(\beta)<k<\ln K,\ \ \ k'=\frac{\ln K-k}{2}$
and split the integral in question as follows
\begin{eqnarray}\label{3.17}
&\int_{h}^{\sqrt{h}}\frac{ds}{\sqrt{s}}\int_{c_1(s)}^{c_2(s) }dz|
\frac{\partial^2 u }{\partial t^2}(s,z)|\\
&=\int_{h}^{\sqrt{h}}\frac{ds}{\sqrt{s}}\int_{c_1(s)}^{k'-\mu s}dz|\frac{
\partial^2 u }{\partial t^2}(s,z)|+\int_{h}^{\sqrt{h}}\frac{ds}{\sqrt{s}}
\int_{k'-\mu s }^{c_2(s) }dz|\frac{\partial^2 u }{\partial t^2}(s,z)|.
\nonumber\end{eqnarray}
From Corollary \ref{cor2.8} we know that the function $\frac{\partial^2 u }
{\partial t^2}(s,z)$ is in  $L^2[\tilde E],$ where
$$
\tilde E=\{(s,z): 0<s<T, c_1(t)<z< k'-\mu t\}\subset E^{\sigma}
$$
(for an appropriate $b<\ln K$ in the definition of $E^{\sigma}$).
Therefore we can use the Cauchy-Schwarz inequality to obtain
\begin{eqnarray}\label{3.18}
&\int_{h}^{\sqrt{h}}\frac{ds}{\sqrt{s}}\int_{c_1(s)}^{k'-\mu s}dz|
\frac{\partial^2 u }{\partial t^2}(s,z)|\\
&\leq\big(\int_{h}^{\sqrt{h}}\frac{ds}{s}\int_{c_1(s)}^{k'-\mu s}dz\big)^{1/2}
\big(\int_{h}^{\sqrt{h}}\int_{c_1(s)}^{k'-\mu s}|\frac{\partial^2 u }
{\partial t^2}(s,z)|^2dz\big)^{1/2}\leq C\ln n.\nonumber
\end{eqnarray}

Now we are left with the second integral in the right hand side of 
(\ref{3.17}). We will show that there is a constant $C>0$ such that,
\begin{equation}\label{lamstat3.5}
\cI_n= \int_{h}^{\sqrt{h}}\frac{ds}{\sqrt{s}}\int_{k'-\mu s}^{c_2(s) }dz|
\frac{\partial^2 u }{\partial t^2}(s,z)|\leq Cn^{1/4}.
\end{equation}
Recall that $u(t,x)=e^{-rt}P(t,x+\mu t)$, and so
$$
\frac{\partial^2 u}{\partial t^2}(t,x)=
e^{-rt}\big(r^2P(t,x+\mu t)-2rP_{t}(t,x+\mu t)-2r\mu P_{x}(t,x+\mu t)\big)
$$
$$
+e^{-rt}\big(\mu^2P_{xx}(t,x+\mu t)+2\mu P_{xt}(t,x+\mu t)+P_{tt}(t,x+\mu t) 
\big).
$$
Observe that the functions $P(t,x),P_x(t,x),P_t(t,x)$ and $P_{xx}(t,x)$ are
all bounded for small $t$. Indeed, $P\leq K+\delta$ while $P_t$ is bounded
in the domain of integration in (\ref{lamstat3.5}) for small $h$ in view of
(\ref{P=w+v}), (\ref{w and v}), (\ref{P_t=w+v}) and (\ref{them2.5(1)}). 
Next, $P_x$ is bounded by Theorem 8.1 from \cite{Ku}. Finally, $P_{xx}$
is bounded since in the domain in question $P$ and its first derivatives
are bounded and $P$ satisfies the equation
$(\frac {\partial}{\partial t}+\bfA)P=0$ (see (\ref{P,A})). Therefore, we 
can write
\begin{equation}\label{I_n}
\cI_n \leq \int_{h}^{\sqrt{h}}\frac{dt}{\sqrt{t}}\int_{k'-\mu t}^{c_2(t)}
dx\big(|2\mu e^{-rt}P_{tx}(t,x+\mu t)|+|e^{-rt}P_{tt}(t,x+\mu t)|\big)+C_1,
\end{equation}
for some constant $C_1>0$ independent of $n$. Recall that
for $(x,t)\in D=(0,\beta)\times(k,\ln K)$ by Proposition \ref{thm2.5},
$P_{t}(t,x)=v(t,x)+w(t,x)$
where  $v_t$ and $v_x$ belong to $L^2[D]$. Hence, expressing
$P_{tx}$ and $P_{tt}$ via $v_t,w_t$ and $v_x,w_x$ we can estimate 
the integral (\ref{I_n}) containing $v_t$ and $v_x$ by means of the 
Cauchy-Schwarz inequality as it was done in (\ref{3.18}). Replacing 
these integrals by $C_2\sqrt {\ln n}$ we obtain
$$
\cI_n \leq \int_{h}^{\sqrt{h}}\frac{dt}{\sqrt{t}}\int_{k'-\mu t }^{c_2(t)}dx
\big(|2\mu e^{-rt}w_{x}(t,x+\mu t)|+|e^{-rt}w_{t}(t,x+\mu t)|\big)+C_2
\sqrt {\ln n}+C_1.
$$
By (\ref{w_t>0}) and (\ref{w_x>0}) the functions $w_t(t,x)$ and $w_x(t,x)$ 
do not change signs in $D$, and so it follows that
\begin{eqnarray}\label{3.19}
&\int_{h}^{\sqrt{h}}\frac{dt}{\sqrt{t}}\int_{k'-\mu t}^{c_2(t)}\big(|2\mu
 e^{-rt}w_{x}(t,x+\mu t)|+|e^{-rt}w_{t}(t,x+\mu t)|\big)dx\\
&=\Big{|}\int_{h}^{\sqrt{h}}\frac{dt}{\sqrt{t}}2|\mu| e^{-rt}\int_{k'}^{\ln K
-\kappa\sqrt{h}}w_{x}(t,x)dx\Big{|}+\Big{ |}\int_{h}^{\sqrt{h}}\frac{dt}{
\sqrt{t}}\int_{k'-\mu t}^{c_2(t)}e^{-rt}w_{t}(t,x+\mu t)dx\Big{|}.\nonumber
\end{eqnarray}
By Proposition \ref{thm2.5},  $w(x,t)$ is bounded on $D$, 
and so the contribution of the first integral in the right hand side of 
(\ref{3.19}) is bounded by a constant and it remains to estimate only 
the second integral there.
%%%%%%%%%%%%%%%%%%%%%%%%%%%%%%%%%%%%%function z%%%%%%%%%%%%%%%%%%%%%%%%%%%%%%%%%%%%%%

Next, we will need a more explicit representation of the function $w$. Let
\begin{equation}\label{z_tilde}
\tilde z(t,x)=e^{-rt}w(t,x+\mu t).
\end{equation}
Then in the domain $\tilde E=\{(t,x),0< t< \beta, k-\mu t<x<\ln K-\mu t \}$,
$$
\frac{\kappa^2}{2}\tilde z_{xx}(t,x)+\tilde z_{t}(t,x)=0.
$$
Define
\begin{equation}\label{z}
z(t,x)=\tilde z(T-t,\frac{\kappa}{\sqrt{2}}x)
\end{equation}
and let 
$$
E=\{(t,x): 0<t<T,\  \frac{\sqrt{2}(k-\mu t)}{\kappa}< x< 
\frac{\sqrt{2}(\ln K-\mu t)}{\kappa} \}.
$$
In the domain $E$ the function $z(t,x)$ satisfies the heat equation
$$
z_{xx}(t,x)=z_t(t,x).
$$
If we let
\begin{equation}\label{2.9}
d_1(t)=\frac{\sqrt{2}(k-\mu(T-t))}{\kappa},\ \  d_2(t)=\frac{\sqrt{2}(\ln K
-\mu(T-t))}{\kappa}
\end{equation}
then from the boundary values of $w(t,x)$ we obtain
\begin{equation*}
z(0,x)=0\ \ for\ \    d_1(0)<x<d_2(0),\ \ z(t,d_1(t))=0\,\,\mbox{and}\,\,
z(t,d_2(t))=e^{-r(T-t)}\ \ for \ \ 0<t\leq T.
\end{equation*}
Note that $z(t,x)$ is a bounded continuous function on the boundaries 
$(t,d_i(t)),\ i=1,2\ ,\ 0<t\leq T $ of $E$. Hence, by  Chapter 14 of ~\cite{C} 
we can represent $z(t,x)$ in the form
\begin{equation}\label{z-function}
z(t,x)=\int_0^t\frac{\partial H}{\partial x}(x-d_1(\tau),t-\tau)
\phi_1(\tau)d\tau
+\int_0^t\frac{\partial H}{\partial x}(x-d_2(\tau),t-\tau)\phi_2(\tau)d\tau
\end{equation}
where $H(t,x)=\frac{1}{\sqrt{4\pi t}}e^{-\frac{x^2}{4t}}$ is the fundamental
 solution and the functions $\phi_i(t),\ i=1,2$ are bounded continuous on the 
 interval $(0,T]$.
%%%%%%%%%%%%%%%%%%%%%%%%%%%%%%%%%%%%%function z%%%%%%%%%%%%%%%%%%%%%%%%%%%%%%%%%%%%%%
From the definition of $\tilde z$ we see that
$$
\tilde z_t(t,x)=-re^{-rt}w(t,x+\mu t)+e^{-rt}w_t(t,x+\mu t).
$$
Since $w(t,x)$ is bounded then for some constant $C_1>0$ independent of $n$,
$$
\Big{|}\int_{h}^{\sqrt{h}}\frac{dt}{\sqrt{t}}\int_{k'-\mu t}^{c_2(t)}e^{-rt}
w_{t}(t,x+\mu t)dx\Big{|}\leq
\Big{|}\int_{h}^{\sqrt{h}}\frac{dt}{\sqrt{t}}\int_{k'-\mu t}^{c_2(t)}z_{t}
(t,x)dx\Big{|}+C_1.
$$
From the representation (\ref{z-function}) of $z(t,x)$ we obtain that
\begin{eqnarray}\label{3.20}
&\Big{|}\int_{h}^{\sqrt{h}}\frac{dt}{\sqrt{t}}\int_{k'-\mu t}^{c_2(t)}z_{t}
(t,x)dx\Big{|}\leq\\
&\frac{\kappa}{\sqrt{2}}\Big{|}\int_{h}^{\sqrt{h}}\frac{dt}{\sqrt{t}}\int_{
\frac{\sqrt{2}}{\kappa}(k'-\mu t)}^{\frac{\sqrt{2}}{\kappa}c_2(t)}\frac{d}{dt}
\int_0^{T-t}\frac{\partial H}{\partial x}(x-d_1(\tau),T-t-\tau)\phi_1(\tau)
d\tau dx \Big{|}\nonumber\\
&+\frac{\kappa}{\sqrt{2}}\Big{|}\int_{h}^{\sqrt{h}}\frac{dt}{\sqrt{t}}
\int_{\frac{\sqrt{2}}{\kappa}(k'-\mu t)}^{\frac{\sqrt{2}}{\kappa}c_2(t)}
\frac{d}{dt}\int_0^{T-t}\frac{\partial H}{\partial x}(x-d_2(\tau),T-t-\tau)
\phi_2(\tau)d\tau  dx\Big{|}.\nonumber
\end{eqnarray}

Observe that as long as we keep $x$ or $t$ away from $0$ the function $H(x,t)$ 
is smooth and it has bounded derivatives with bounds depending on the range of
 $t,x$ and their distance from zero. Next, if $x$ satisfies
$$
\frac{\sqrt{2}}{\kappa}(k'-\mu t)<x<\frac{\sqrt{2}}{\kappa}c_2(t)=
\frac{\sqrt{2}}{\kappa}(\ln K-\mu t-\kappa\sqrt{h})
$$
 then  
$$
k'-k\leq\frac{\sqrt{2}}{\kappa}\big(k'-k +\mu(T-t-\tau)\big)<x-d_1(\tau)\ 
\ for\ \ 0<\tau\leq T-t.
$$
Since $k'>k$ we see that $x-d_1(\tau)$ stays away from $0$ on the entire 
interval $(0,T-t]$. It follows from the above that the function
$$
\Phi_1(t,x)= \int_0^{T-t}\frac{\partial H}{\partial x}(x-d_1(\tau),
T-t-\tau)\phi_1(\tau)d\tau
$$
has bounded derivatives with respect to $t$ with bounds independent of $n$ in
 the region $\{(t,x): h<t<\sqrt{h},\frac{\sqrt{2}}{\kappa}(k'-\mu t) 
 <x\frac{\sqrt{2}}{\kappa}c_2(t)\}$. We conclude that the first integral in 
 the right hand side of (\ref{3.20}) is bounded from above by a constant
 independent of $n$ and it remains to estimate the second integral there.
 
Set 
$$
\Phi_2(t,x)=\int_0^{T-t}\frac{\partial H}{\partial x}(x-d_2(\tau),T-t-\tau)
\phi_2(\tau)d\tau.
$$
We see that if 
$$
x<\frac{\sqrt{2}}{\kappa}c_2(t)=\frac{\sqrt{2}}{\kappa}
(\ln K-\mu t-\sqrt{h}),
$$
then 
$$
x-d_2(\tau)=x-\frac{\sqrt{2}}{\kappa}(\ln K-\mu(T-\tau))<\frac{\sqrt{2}}
{\kappa}(\mu(T-t-\tau)-\sqrt{h}).
$$
In this case $x-d_2(\tau)$ can be zero when $\tau \in [0,T-t]$ but
 this can happen only for a $\tau$ that which is at least 
 $\mu^{-1}\sqrt h$ apart from $T-t$. Thus, the function $\Phi_2$ is smooth 
 with a bounded uniformly continuous derivative with respect to $t$ though
 this  bound may depend on $n$. Nevertheless, we still have the following
\begin{eqnarray*}
&\frac{\kappa}{\sqrt{2}}\Big{|}\int_{h}^{\sqrt{h}}\frac{dt}{\sqrt{t}}
\int_{\frac{\sqrt{2}}{\kappa}(k'-\mu t)}^{\frac{\sqrt{2}}{\kappa}c_2(t)}
\frac{d}{dt}\big(\int_0^{T-t}\frac{\partial H}{\partial x}(x-d_2(\tau),
T-t-\tau)\phi_2(\tau)d\tau\big)dx\Big{|}\\
&=\frac{\kappa}{\sqrt{2}}\Big{|}\int_{h}^{\sqrt{h}}\frac{dt}{\sqrt{t}}
\frac{d}{dt}\big(\int_0^{T-t}\int_{\frac{\sqrt{2}}{\kappa}(k'-\mu t)}^{
\frac{\sqrt{2}}{\kappa}c_2(t)}\frac{\partial H}{\partial x}(x-d_2(\tau),
T-t-\tau)dx\phi_2(\tau)d\tau \big)\Big{|}\\
&=\frac{\kappa}{\sqrt{2}}\Big{|}\int_{h}^{\sqrt{h}}\frac{dt}{\sqrt{t}}
\frac{d}{dt}\Big(\int_0^{T-t}\big(H(\frac{\sqrt{2}}{\kappa}c_2(t)-d_2(\tau),
T-t-\tau)\\
&-H(\frac{\sqrt{2}}{\kappa}(k'-\mu t)-d_2(\tau),
T-t-\tau)\big)\phi_2(\tau)d\tau \Big)\Big{|}\\
&\leq\frac{\kappa}{\sqrt{2}}\Big{|}\int_{h}^{\sqrt{h}}\frac{dt}{\sqrt{t}}
\frac{d}{dt}\Big(\int_0^{T-t}
\big(H(\frac{\sqrt{2}}{\kappa}c_2(t)-d_2(\tau),T-t-\tau)\phi_2(\tau)d\tau 
\Big)\Big{|}\\
&+\frac{\kappa}{\sqrt{2}}\Big{|}\int_{h}^{\sqrt{h}}
\frac{dt}{\sqrt{t}}\frac{d}{dt}\Big(\int_0^{T-t}
H(\frac{\sqrt{2}}{\kappa}(k'-\ln K+\mu(T-t-\tau),T-t-\tau)\phi_2(\tau)
d\tau \Big)\Big{|}.
\end{eqnarray*}
We see that in the second term in the right hand side 
$k'-\ln K+\mu(T-t-\tau)$ can take on the value $0$ for $\tau\in(0,T-t]$ but
then $\tau$ is at least $c=\mu^{-1}|k'-\ln K|$ apart from $T-t$ and now the
separation constant $c$ does not depend on $n$. Thus, we can bound the
second term there from above by a constant and it remains to estimate the
first term which we do  as follows
\begin{eqnarray*}
&\textbf{I} = \frac{\kappa}{\sqrt{2}}\Big{|}\int_{h}^{\sqrt{h}}\frac{dt}
{\sqrt{t}}\frac{d}{dt}\int_0^{T-t}
H(\frac{\sqrt{2}}{\kappa}c_2(t)-d_2(\tau),T-t-\tau)\phi_2(\tau)d\tau \Big{|}\\
&\leq C_2\int_{h}^{\sqrt{h}}\frac{dt}{\sqrt{t}}\int_0^{T-t}
\Big{|}\frac{1}{(T-t-\tau)^{3/2}}\exp\big(-\frac{(\frac{\sqrt{2}}{\kappa}
(\mu(T-t-\tau)-\kappa\sqrt{h})^2}{4(T-t-\tau)}\big)
\phi_2(\tau)\Big{|}d\tau\\
&+C_2\int_{h}^{\sqrt{h}}\frac{dt}{\sqrt{t}}\int_0^{T-t}
\Big{|}\frac{1}{(T-t-\tau)^{1/2}}
\exp\big(-\frac{(\frac{\sqrt{2}}{\kappa}(\mu(T-t-\tau)-\kappa\sqrt{h})^2}
{4(T-t-\tau)}\big)\phi_2(\tau)\Big{|}d\tau\\
&+C_2\int_{h}^{\sqrt{h}}\frac{dt}{\sqrt{t}}\int_0^{T-t}
\Big{|}\frac{h}{(T-t-\tau)^{5/2}}
\exp\big(-\frac{\ka^2h}{2(T-t-\tau)}\big)\exp\big(\frac{\mu\sqrt{h}}{\sqrt{2}}
-\frac{\mu^2}{2\kappa^2}(T-t-\tau) \big)\phi_2(\tau)\Big{|}d\tau
\end{eqnarray*}
where $C_2>0$ is a constant independent of $n$.
Analyzing the integral with respect to $\tau$ in the second term in the
right hand side above by considering different possible values of $T-t-\tau$
we conclude that this integral is bounded by a constant independent of $n$. 
Next we observe that $|\exp\big(\frac{\mu\sqrt{h}}{\sqrt{2}}-\frac{\mu^2}
{2\kappa^2}(T-t-\tau)\big)\phi_2(\tau)|$ is also bounded by a constant 
independent of $n$ too. Hence, we obtain
\begin{eqnarray*}
&\textbf{I}\leq C_3+ C_3\int_{h}^{\sqrt{h}}\frac{dt}{\sqrt{t}}\int_0^{T-t}
\frac{1}{(T-t-\tau)^{3/2}}\exp\big(-\frac{h}{2(T-t-\tau)}\big)
d\tau\\
&+C_3\int_{h}^{\sqrt{h}}\frac{dt}{\sqrt{t}}\int_0^{T-t}
\frac{h}{(T-t-\tau)^{5/2}}
\exp\big(-\frac{h}{2(T-t-\tau)}\big)d\tau
\end{eqnarray*}
for a constant $C_3>0$ independent of $n$.
Set $\rho=\sqrt{\frac{h}{2(T-t-\tau)}}$ and note that $\frac{d\rho}{d\tau}
=-\frac{\sqrt{h}}{2\sqrt 2(T-t-\tau)^{3/2}}$ and $\frac{d\rho^2}{d\tau}=
-\frac{\sqrt{h}}{4(T-t-\tau)^{2}}$. We proceed by changing variables 
arriving at
\begin{eqnarray*}
&\textbf{I}\leq C_4+C_4\int_{h}^{\sqrt{h}}\frac{dt}{\sqrt{t}}\int_{\frac{
\sqrt{h}}{\sqrt{2(T-t)}}}^{\infty}
\frac{1}{\sqrt{h}}e^{-\rho^2}
d\rho+C_4\int_{h}^{\sqrt{h}}\frac{dt}{\sqrt{t}}\int_{\frac{h}{2(T-t)}}^{\infty}
\frac{1}{\sqrt{h}}\rho e^{-\rho^2}
d\rho^2\\
&\leq C_4+C_5\frac{1}{\sqrt{h}}\int_{h}^{\sqrt{h}}\frac{dt}{\sqrt{t}}
\leq C_4+C_52(1+\frac{1}{h^{1/4}})\leq C_6n^{1/4}
\end{eqnarray*}
for some constants $C_4,C_5,C_6>0$ independent of $n$ and (\ref{lamstat3.5})
follows. Combining (\ref{lamstat3.5}) and (\ref{3.18}) we obtain from 
(\ref{3.17}) that
\begin{equation}\label{this}
\int_{h}^{\sqrt{h}}\frac{ds}{\sqrt{s}}\int_{c_2(s)}^{c_1(s)}dz|
\frac{\partial^2u}{\partial t^2}(s,z)|\leq Cn^{1/4}.
\end{equation}
Finally, Proposition \ref{prop 3.1} follows from (\ref{this}), (\ref{3.15}) 
and (\ref{3.16}).
\end{proof}

Next, we turn our attention to the domain $\textbf{B}$. First, we will
prove the following result.
\begin{lem}\label{lem 3.7}
There exists a constant $C>0$ such that for all $n\in\bbN$,
\begin{equation}\label{3.21}
\bfE [\sum_{j=1}^{h^{-1}(\tau\wedge\sig^{(n)})}\cD u((j-1)h,X_{(j-1)h}^{(n)})
\bbI_{((j-1)h,
X_{(j-1)h})|\in\textbf{B}}]\leq Cn^{-3/4}.
\end{equation}
\end{lem}
\begin{proof}
Let $\textbf{B}_{t<\beta^{(n)}}$ and $\textbf{B}_{t\geq\beta^{(n)}}$ be the 
set of all points  $(t,x)\in\textbf{B}$ such that $t<\beta^{(n)}$ and $t\geq 
\beta^{(n)}$, respectively.
We split (\ref{3.21}) according to these two regions, namely,
\begin{eqnarray*}
&\bfE [\sum_{j=1}^{h^{-1}(\tau\wedge\sig^{(n)})}\cD u((j-1)h,X_{(j-1)h}^{(n)})
\bbI_{((j-1)h,X_{(j-1)h})|\in\textbf{B}}]\\
&=\bfE [\sum_{j=1}^{h^{-1}(\tau\wedge\sig^{(n)})}\cD u((j-1)h,X_{(j-1)h}^{(n)})
\bbI_{((j-1)h,X_{(j-1)h})|\in\textbf{B}_{t<\beta^{(n)}}}]\\
&+\bfE [\sum_{j=1}^{h^{-1}(\tau\wedge\sig^{(n)})}\cD u((j-1)h,X_{(j-1)h}^{(n)})
\bbI_{((j-1)h,X_{(j-1)h})|\in\textbf{B}_{t\geq\beta^{(n)}}}].
\end{eqnarray*}
By Proposition \ref{prop 3.1} we have that for a constant $C>0$ independent
of $n$,
$$
\bfE [\sum_{j=1}^{h^{-1}(\tau\wedge\sig^{(n)})}\cD u((j-1)h,X_{(j-1)h}^{(n)})
\bbI_{((j-1)h,X_{(j-1)h})|\in\textbf{B}_{t\geq\beta^{(n)}}}]\leq Cn^{-3/4}.
$$
Thus, it remains to estimate only the first term in the right hand side. Let 
$E=\{(t,x):0<t<\beta^{(n)},a-\mu t<x<b-\mu t\}$ where $a< s(0)$ and 
$s(\beta^{(n)}+h)+|\mu|h+2\kappa\sqrt{h}<b<\ln K$. For $n$ large enough we
 can find such a $b$ because $s(t)$ is continuous and $s(\beta^{(n)})<\ln K$.
We know from Corollary \ref{cor2.8} that $u(t,x)\in H^2[E]$. Since 
$C^2[E]$ is dense in this space we can approximate $u(t,x)$ by $C^2$ 
functions to get equality (\ref{difoper u}) of Proposition 
\ref{dis_operator_prop 2.2} for $u(t,x)$, as well. Since 
$u_t(t,x)+\frac{\kappa^2}{2}u_{xx}(t,x)\leq 0$ in the domain $E$ we obtain
\begin{eqnarray*}
&\cD u(t,x)\leq\frac{1}{\kappa}\int_{0}^{\sqrt{h}}dy\int_{-\kappa y}^{
\kappa y}dz\big (z\frac{\partial^2 u }{\partial t\partial x}(t+y^2,x+z) 
\big)\\
&\leq \int_{0}^{\sqrt{h}}ydy\int_{-\kappa \sqrt{y}}^{\kappa \sqrt{y}}dz 
\big |\frac{\partial^2 u }{\partial t\partial x}(t+y^2,x+z) \big|\\
&=\frac{1}{2}\int_{0}^{{h}}ds\int_{-\kappa \sqrt{y}}^{\kappa 
\sqrt{y}}dz \big |\frac{\partial^2 u }{\partial t\partial x}(t+s,x+z) \big|.
\end{eqnarray*}
It follows that
$$
\cD u(t,y)\bbI_{(t,y)\in \bar B}\leq \frac{1}{2}\int_{t}^{t+h}ds\int_{s(t)-
\lambda\sqrt{h}-\mu t}^{s(t+h)+\lambda\sqrt{h}-\mu t}\bbI_{|z-y|\leq 
\kappa\sqrt{h}}\big |\frac{\partial^2 u }{\partial t\partial x}(s,z) \big|dz
$$
where $\lambda =|\mu|+\kappa$. Hence,
\begin{eqnarray*}
&\bfE [\sum_{j=1}^{h^{-1}(\tau\wedge\sig^{(n)})}\cD u((j-1)h,X_{(j-1)h}^{(n)})
\bbI_{\{((j-1)h,X_{(j-1)h}))
\in\textbf{B}_{t<\beta^{(n)}}\}}]\\
&\leq\frac{1}{2}\big{ (}\sum_{j=1}^{k_{\beta}} \int_{jh}^{(j+1)h}d\tau
\int_{s(jh)-\lambda\sqrt{h}-\mu jh}^{s(jh+h)+\lambda\sqrt{h}-\mu jh}
\textbf{P}\big(|X^{(n)}_{jh}-z|\leq\kappa\sqrt{h}\big) \big |
\frac{\partial^2 u }{\partial t\partial x}(s,z) \big|dz  \big{ )}+\frac{C}{n}.
\end{eqnarray*}
Here $k_{\beta}=\lceil \frac{\beta}{h} \rceil$,
and the term $\frac{C}{n}$ is the contribution of $\cD u(0,X^{n}_0)=\cD
 u(0,x)\leq \frac{C}{n}$ which holds true from by the definition of the
operator $\cD$ and boundedness of $u_t$ and $u_{xx}$ for small $t$.
From Corollary \ref{cor2.8} we see that there exists a constant $C_1>0$
such that
\begin{equation}
\int_{a}^b|\frac{\partial^2 u }{\partial t\partial x}(t,z)\big|^2dz\leq C_1\,\,
\mbox{when}\,\, 0\leq t\leq \beta^{(n)}.
\end{equation}
This together with (\ref{Berry-Esseen estimate}), the Cauchy-Schwarz inequality 
and the inequality $\frac{1}{\sqrt{\tau}}\geq  \frac{1}{\sqrt{2jh}}$, which
is satisfied when $j\geq 1$ and $jh\leq\tau \leq 2jh $, yields that 
\begin{eqnarray*}
&\bfE [\sum_{j=1}^{h^{-1}(\tau\wedge\beta^{(n)})}\cD u((j-1)h,X_{(j-1)h}^{(n)})
\bbI_{\{((j-1)h,X_{(j-1)h}))\in\textbf{B}_{t<\beta^{(n)}}\}}]\\
&\leq \sqrt{h}C_2\sum_{j=1}^{k_{\beta}}\int_{jh}^{(j+1)h}\frac{d\tau}
{\sqrt{\tau}}\big((s(j+1)h)-s(jh)+2\lambda\sqrt{h}\big)^{1/2}+\frac{C_2}{n} 
\end{eqnarray*}

From Proposition \ref{prop2.4} and Lipschitz continuity of the function 
$P(t,x)$ in $t\leq \beta^{(n)}$ uniformly in $x\leq\ln K$ (see Theorem 8.1
in \cite{Ku}) we obtain that for some constant $C_3>0$,
$$
|s(t_1)-s(t_2)|\leq \sqrt{|t_1-t_2|}C_3\,\,\mbox{whenever}\,\, 0\leq t_1,t_2
\leq\beta^{(n)}.
$$
Hence,
$$
\bfE [\sum_{j=0}^{h^{-1}(\tau\wedge\beta^{(n)})}\cD u((j-1)h,X_{(j-1)h}^{(n)})
\bbI_{\{((j-1)h,X_{(j-1)h}))
\in\textbf{B}_{t<\beta^{(n)}}\}}]\leq \frac{C_4}{n^{3/4}}
$$
for some constant $C_4>0$ independent of $n$.
\end{proof}
By combining the results of Lemma \ref{lem 3.2} , Proposition \ref{prop 3.3}
 and Lemma \ref{lem 3.7} together with (\ref{3.8}) we obtain that the upper
 bound $P_1^{(n)}(x)-P(0,x)<\frac{C}{n^{3/4}}$ for some constant $C>0$ 
 independent of $n$ and of $x\leq \ln K$.

Next, we will obtain a lower bound for the approximation error 
$P_1^{(n)}(x)-P(0,x)$ when $x \leq \ln K$. Set
\begin{equation}\label{3.23}
\tau^{(n)}=inf\{t:\mu [t/h]h+X^{(n)}_{t}<s([t/h]h+h)+|\mu|h+\kappa\sqrt{h}\}.
\end{equation}
By Proposition \ref{prop 3.3},
\begin{eqnarray}\label{3.24}
&\bfE [u(\tau^{(n)}\wedge\sigma^{(n)},X^{(n)}_{\tau^{(n)}\wedge\sigma^{(n)}})]
=u(0,x)+\bfE [\sum_{j=1}^{\tau^{(n)}\wedge\sigma^{(n)}/h}\cD 
u((j-1)h,X_{(j-1)h}^{(n)})]\\
&=u(0,x)+\bfE [\sum_{j=1}^{h^{-1}(\tau\wedge\sig^{(n)})}|\cD u((j-1)h,X_{(j-1)h}^{(n)})
\bbI_{\{((j-1)h,X_{(j-1)h\}})|\in\textbf{C}}]\geq P(0,x)-Cn^{-3/4}.
\nonumber\end{eqnarray}

Set $\alpha=\alpha_n=T-\frac{1}{n^{2/3}}$ and let $\tau_A^{(n)}$ be defined
by (\ref{3.23}) with $s$ there replaced by the free boundary $s_A$ for the
 American put option (see Section 2.2 in \cite{L}). 
 Define also $\tau_{\alpha}^{(n)}=\tau^{(n)}\bbI_{\{\tau^{(n)}+h<\alpha\}}
 +T\bbI_{\{\tau^{(n)}+h\geq \alpha\}}$ and $\tau_{A,\alpha}^{(n)}=
 \tau_A^{(n)}\bbI_{\{\tau_A^{(n)}+h<\alpha\}}+T\bbI_{\{\tau_A^{(n)}+h\geq 
 \alpha\}}$. We will rely on the following estimate from Section 4.5 
 in \cite{L}.

\begin{lem}\label{lem 3.8}
 There exists a constant $C>0$ independent of $n\bbN$ such that
\begin{equation}\label{3.25}
|\bfE [u_A(\tau_A^{(n)},X^{(n)}_{\tau_A^{(n)}})-e^{r\tau^{(n)}_{A,\alpha}}
\psi(\mu\tau^{(n)}_{A,\alpha}+X_{\tau^{(n)}_{A,\alpha}})]|\leq 
\frac {C}{n^{2/3}}
\end{equation}
where $u_A(t,x)=e^{-tr}P_A(t,x+\mu t)$ with $P_A$ given by (\ref{2.-3}).
\end{lem}
\begin{rem}\label{rem3.9}
Note that $s_A(t)=s(t)$ for $\beta\leq t<T$, and so $\tau_A^{(n)}\vee\beta
=\tau^{(n)}\vee\beta$ and $\tau^{(n)}_{A,\alpha}\vee\beta=
\tau^{(n)}_{\alpha}\vee\beta$.
\end{rem}

From now on we assume that $n$ is large enough so that $\beta^{(n)}<\alpha$.
From the definition of $P^{(n)}_1(x)$ we have
\begin{equation}\label{3.26}
P^{(n)}_1(x)\geq \bfE [e^{-r\tau_\alpha^{(n)}\wedge\sigma^{(n)}}\big(\psi(
\mu\tau_\alpha^{(n)}+X^{(n)}_{\tau_\alpha^{(n)}} )\bbI_{\{\tau_\alpha^{(n)}
\leq\sigma^{(n)}\}}+\delta\bbI_{\{\sigma^{(n)}<\tau_\alpha^{(n)}\}} \big)].
\end{equation}
Hence, if we prove that for some constant $C>0$ independent of $n$,
\begin{equation}\label{3.27}
\bfJ=|\bfE [u(\tau^{(n)}\wedge\sigma^{(n)},X^{(n)}_{\tau^{(n)}\wedge
\sigma^{(n)}})]-\bfE [e^{-r\tau_\alpha^{(n)}\wedge\sigma^{(n)}}\big(\psi(\mu
\tau_\alpha^{(n)}+X^{(n)}_{\tau_\alpha^{(n)}} )\bbI_{\{\tau_\alpha^{(n)}\leq\sigma^{(n)}\}}
+\delta\bbI_{\sigma^{(n)}<\tau_\alpha^{(n)}} \big)]|\leq \frac{C}{\sqrt{n}}
\end{equation}
then by (\ref{3.26}) and (\ref{3.24}) we could conclude that
\begin{equation}\label{3.28}
-\frac{C}{\sqrt{n}}\leq P^{(n)}_1(x)-P(0,x).
\end{equation}
We split the left hand side of (\ref{3.27}) into three parts
\begin{eqnarray}\label{3.29}
&\bfJ=\bfE [\big\{u(\tau^{(n)}\wedge \beta^{(n)} ,X^{(n)}_{\tau^{(n)}\wedge
\beta^{(n)}})-e^{-r\tau^{(n)}\wedge\beta^{(n)}}\big(\psi(\mu\tau^{(n)}\wedge
\beta^{(n)}\\
&+X^{(n)}_{\tau^{(n)}\wedge\beta^{(n)}} ) \big)\big\}\bbI_{\{\tau_\alpha^{(n)}\leq
 \sigma^{(n)}\wedge\beta^{(n)}\}}]
+ \bfE [\big\{u(\sigma^{(n)},X^{(n)}_{\sigma^{(n)}})-e^{-r\sigma^{(n)}}\delta
\big\}\bbI_{\{\sigma^{(n)}<\tau_\alpha^{(n)}\}}]\nonumber\\
&+\bfE [\big\{u(\tau^{(n)} ,X^{(n)}_{\tau^{(n)}})-e^{-r\tau_\alpha^{(n)}}
\big(\psi(\mu\tau_\alpha^{(n)}+X^{(n)}_{\tau_\alpha^{(n)}} ) \big)\big
\}\bbI_{\{\beta^{(n)}<\tau_\alpha^{(n)}\leq \sigma^{(n)}\}}].\nonumber
\end{eqnarray}
This equality is true since $\tau^{(n)}=\tau_\alpha^{(n)}=\tau^{(n)}\wedge\beta$ 
on the set $\tau_\alpha^{(n)}\leq \sigma^{(n)}\wedge \beta^{(n)}<\alpha$. We begin 
with the last term. First note that on the set $\beta^{(n)}
\leq \tau_\alpha^{(n)}\leq \sigma^{(n)}$ we have, in particular, 
$\beta^{(n)}\leq \sigma^{(n)}$ and so $\sigma^{(n)}=T$ by Remark 
\ref{rem3.9}. In the case  $\tau_\alpha^{(n)}>\beta^{(n)}$ we have 
$\tau_\alpha^{(n)}=\tau_{A,\alpha}^{(n)}$ and $\tau^{(n)}=\tau_A^{(n)}$ and 
so from Lemma \ref{lem 3.8} we derive that
\begin{eqnarray*}
&|\bfE [\big\{u(\tau^{(n)} ,X^{(n)}_{\tau^{(n)}})-e^{-r\tau_\alpha^{(n)}}
\big(\psi(
\mu\tau_\alpha^{(n)}+X^{(n)}_{\tau_\alpha^{(n)}} ) \big)\big\}
\bbI_{\{\beta^{(n)}<\tau_{\alpha}^{(n)}\leq \sigma^{(n)}\}}]|\\
&\leq |\bfE [\big\{u(\tau_A^{(n)} ,X^{(n)}_{
{\tau}^{(n)}_A})-e^{-r{\tau}_{A,\alpha}^{(n)}}\big(\psi(\mu
{\tau}_{A,\alpha}^{(n)}+X^{(n)}_{
{\tau}_{A,\alpha}^{(n)}} ) \big)\big\}]|\leq \frac{C}{n^{2/3}}.
\end{eqnarray*}

Next, we deal with the first term in the right hand side of (\ref{3.29}) 
where  $\tau_\alpha^{(n)}=\tau^{(n)}\leq\sigma^{(n)}\wedge\beta^{(n)}$.
This means that before time $\beta^{(n)}$ the process $X^{(n)}$ is
stopped near the boundary $s(t)$ and 
$$
\mu\tau^{(n)}+X^{(n)}_{\tau^{(n)}}<s(\tau^{(n)}+h)+|\mu|h+
\sigma\sqrt{h}.
$$
By the definition, $u(\tau^{(n)},X^{(n)}_{\tau^{(n)}})=e^{-r\tau^{(n)}}
P(\tau^{(n)},X^{(n)}+\mu\tau^{(n)})$.
Thus, we have
\begin{eqnarray*}
&\bfE [\big\{u(\tau^{(n)}\wedge \beta^{(n)} ,X^{(n)}_{\tau^{(n)}})-e^{-
r\tau^{(n)}\wedge\beta^{(n)}}\big(\psi(\mu\tau^{(n)}\wedge\beta^{(n)}+
X^{(n)}_{\tau^{(n)}\wedge\beta^{(n)}} ) \big)\big\}\bbI_{\{\tau_\alpha^{(n)}
\leq \sigma^{(n)}\wedge\beta^{(n)}\}}]\\
&=\bfE [\big\{e^{-r\tau^{(n)}}\big(P(\tau^{(n)},X^{(n)}+\mu\tau^{(n)})-
\psi(\mu\tau^{(n)}+X^{(n)}_{\tau^{(n)}} ) \big)\big\}\bbI_{\{\tau_\alpha^{(n)}
\leq \sigma^{(n)}\wedge\beta^{(n)}\}}].
\end{eqnarray*}
If $\mu\tau^{(n)}+X^{(n)}_{\tau^{(n)}}\leq s(\tau^{(n)})$ then $P(\tau^{(n)},
X^{(n)}+\mu\tau^{(n)})-\psi(\mu\tau^{(n)}+X^{(n)}_{\tau^{(n)}} )=0$ so we can
 assume that
\begin{equation}\label{3.30_err_es}
s(\tau^{(n)})<\mu\tau^{(n)}+X^{(n)}_{\tau^{(n)}}<s(\tau^{(n)}+h)+|\mu|h+
\sigma\sqrt{h}.
\end{equation}
To continue we need the following lemma.
\begin{lem}\label{lem 3.10}
There is a constant $C>0$ independent of $n$ such that for every point 
$(t,x)$  satisfying $s(t)\leq \mu t+x\leq s(t+h)+|\mu|h+\sigma\sqrt{h}$ 
and $0\leq t\leq \beta^{(n)}$,
$$
|P(t,\mu t+x)-\psi(\mu t+x)|\leq \frac{C}{n}.
$$
\end{lem}
\begin{proof}
The function $P(t,x)$ is Lipschitz continuous when $t\leq \beta$ and
$0\leq x\leq\mu\beta+\ln K$ (see \cite{Ku}), and so
$$
|P(t,\mu t+x)-P(t+h,\mu t+x)|\leq \frac{C}{n}
$$
for some $C>0$ independent of $n$. If $\mu t+x\leq s(t+h)$ then 
$P(t+h,\mu t+x)=\psi(\mu t+x)$ and we are done.
 Now assume that $s(t+h)<\mu t+x<s(t+h)+ \frac{\lambda}{\sqrt{n}}$ where
$\lambda=\sqrt{n}(|\mu|h+\sigma \sqrt{h})$.

From Corollary \ref{cor12} it follows that for every $t<T$ and $a<\ln K$
the function $P_{xx}(t,x)$ is continuous in $x$ on the closed interval 
$[s(t),a]$, so we can write
$$
P(t+h,\mu t+x)=P(t+h,s(t+h))+P_x(t+h,s(t+h))\frac{\lambda}{\sqrt{n}}+P_{xx}(t+h,s(t+h))
\frac{\lambda^2}{2n}+\alpha
$$
where $\al=\al(h)$ satisfies $\lim_{h\to 0}(\frac{\alpha}{h})=0.$
From the property of smooth fit (see ~\cite{Ku}) it follows that
$P_x(t+h,s(t+h))=\psi_x(s(t+h))$,
and so for some $C>0$ independent of $n$,
$$
|P(t+h,\mu t+x)-\psi(\mu t+x)|\leq \frac{C}{n}.
$$
\end{proof}
Using (\ref{3.30_err_es}) and the above lemma we obtain
\begin{eqnarray}\label{3.31}
&|\bfE [\big\{u(\tau^{(n)}\wedge \beta^{(n)} ,X^{(n)}_{\tau^{(n)}\wedge
\beta^{(n)}})-e^{-r\tau^{(n)}\wedge\beta^{(n)}}\big(\psi(\mu\tau^{(n)}\wedge
\beta^{(n)}\\
&+X^{(n)}_{\tau^{(n)}\wedge\beta^{(n)}} ) \big)\big\}\bbI_{\{\tau_\alpha^{(n)}\leq 
\sigma^{(n)}\wedge\beta^{(n)}\}}]|\leq \frac{C}{n}.\nonumber
\end{eqnarray}
Hence, we are done with the first term in the right hand side of (\ref{3.29})
and it remains to estimate the second one.
Since $\sigma^{(n)}<\tau_\alpha^{(n)}\leq T$ the process $X^{(n)}$ is stopped
near the writer's boundary. Namely, we have
$$
\ln K-|\mu|h-\sigma\sqrt h<\mu \sigma^{(n)}+X^{(n)}_{\sigma^{(n)}}\leq \ln K.
$$
Since $P(t,\ln K)=\delta$ when $t\leq \beta$, $\beta^{(n)}-\beta<h$ and
$P$  is Lipschitz  continuous (see Theorem 8.1 of \cite{Ku}) we obtain that  
$$
|P(\sigma^{(n)},\mu\sigma^{(n)}+X^{(n)}_{\sigma^{(n)}})-\delta|\leq 
\frac{C}{\sqrt{n}}
$$
for some $C>0$ independent of $n$. Hence,
\begin{equation}\label{3.32}
\bfE [\big(u(\sigma^{(n)},X^{(n)}_{\sigma^{(n)}})-e^{-r\sigma^{(n)}}\delta
\big)\bbI_{\{\sigma^{(n)}<\tau_\alpha^{(n)}\}}]\leq \frac{C}{\sqrt{n}}.
\end{equation}
It follows that there exists $C>0$ independent of $n$ such that for every 
$x\leq \ln K$,
\begin{equation}\label{P1_low_x<}
-\frac{C}{\sqrt{n}}<P_1^{(n)}(x)-P(0,x).
\end{equation}

%\subsubsection{\textit{Lower bound for} $P^{(n)}_2$}
Next, we will derive a lower bound for the second approximation function 
$P_2^{(n)}(x)$ defined by (\ref{P^n2}), still assuming that $x\leq \ln K$.
According to (\ref{3.24}) in order to obtain
\begin{equation}\label{3.33}
P_2^{(n)}(x)-P(0,x)\geq -\frac{C}{n^{2/3}}.
\end{equation}
it suffices to show that
\begin{equation}\label{3.34}
\bfE [u(\tau^{(n)}\wedge\sigma^{(n)},X^{(n)}_{\tau^{(n)}\wedge\sigma^{(n)}})]
-P_2^{(n)}(x)\leq \frac{C}{n^{2/3}}.
\end{equation}
We have
\begin{eqnarray}\label{3.35}
&\bfE [u(\tau^{(n)}\wedge\sigma^{(n)},X^{(n)}_{\tau^{(n)}\wedge\sigma^{(n)}})]
-P_2^{(n)}(x)\leq\\
&\bfE [u(\tau^{(n)}\wedge\sigma^{(n)},X^{(n)}_{\tau^{(n)}\wedge\sigma^{(n)}})
-e^{-r\tau_\alpha^{(n)}\wedge\sigma^{(n)}}\Big(\psi(\mu\tau_\alpha^{(n)}+
X^{(n)}_{\tau_\alpha^{(n)}} )\bbI_{\{\tau_\alpha^{(n)}\leq \sigma^{(n)}\}}\nonumber\\
&+\big(\psi(\mu\sig^{(n)}+X^{(n)}_{\sig^{(n)}} )+\delta\big)\bbI_{
\{\sigma^{(n)}<\tau_\alpha^{(n)}\}} \Big)]
=\bfE [\big\{u(\tau^{(n)}\wedge \beta^{(n)} ,X^{(n)}_{\tau^{(n)}\wedge\beta^{(n)}})
-e^{-r\tau^{(n)}\wedge\beta^{(n)}}\big(\psi(\mu\tau^{(n)}\wedge\beta^{(n)}
\nonumber\\
&+X^{(n)}_{\tau^{(n)}\wedge\beta^{(n)}} ) \big)\big\}\bbI_{\{\tau_{
\alpha}^{(n)}\leq \sigma^{(n)}\wedge\beta^{(n)}\}}]
+ \bfE [\big\{u(\sigma^{(n)},X^{(n)}_{\sigma^{(n)}})-e^{-r\sigma^{(n)}}(\psi(
\mu\sig^{(n)}+X^{(n)}_{\sig^{(n)}} )+\delta)\big\}\bbI_{\{
\sigma^{(n)}<\tau_\alpha^{(n)}\}}]\nonumber\\
&+\bfE [\big\{u(\tau^{(n)} ,X^{(n)}_{\tau^{(n)}})-e^{-r\tau_\alpha^{(n)}}\big(\psi(
\mu\tau_\alpha^{(n)}+X^{(n)}_{\tau_\alpha^{(n)}} ) \big)\big\}\bbI_{\{\beta^{(n)}
<\tau_\alpha^{(n)}\leq \sigma^{(n)}\}}]\nonumber
\end{eqnarray}
Indeed, the first inequality is true since $P^{(n)}_2(x)$ is defined as the sup
on  $\tau\in\cT^{(n)}$ and we chose a specific one, i.e.  $\tau_\alpha^{(n)}$. 
The equality is true due to the same reason that (\ref{3.29}) holds true. 
We see that the first term in the right hand side of (\ref{3.35}) is the same 
as the first term in (\ref{3.29}) and by (\ref{3.31}) it is less then 
$\frac{C}{n}$ for some constant $C$. The second term is nonpositive because
 for every $(t,x)$ we have $P(t,x)\leq \psi(x)+\delta$ and $u(t,x)=e^{-rt}P(t,\mu t+x)$ so we can just remove it from the equation. The last term is the same as the last term of (\ref{3.29}) and from Lemma (\ref{lem 3.8}) we obtain that this term is less or equal than $\frac{C}{n^{2/3}}$ for an appropriate $C$. 
These arguments yield (\ref{3.34}) and hence  (\ref{3.33}), as well.
For the upper bound we already know that $P^{(n)}_1(x)-P(0,x)\leq 
\frac{C}{n^{3/4}}$ and from the definition of $P^{(n)}_1$ and $P^{(n)}_2$ it 
is not hard to see that $|P^{(n)}_2-P^{(n)}_1|\leq \frac{C}{\sqrt{n}}$. 
It follows from above that there exist $C>0$ such that for every $x\leq \ln K$,
\begin{equation}\label{P2 up_low_x<}
-\frac{C}{n^{3/2}}\leq P_2^{(n)}(x)-P(0,x)\leq \frac{C}{\sqrt{n}}.
\end{equation}

\subsection{Case $x > \ln K$}
%\subsubsection{\textit{Upper bound for} $P^{(n)}_1$}
We begin with the upper bound on $P^{(n)}_1$. We will show first that
\begin{equation}\label{3.36}
P^{(n)}_1(x)-P(0,x)\leq \sup _{\tau\in\cT^{(n)}}\bfE [\sum_{j=1}^{h^{-1}
(\tau\wedge\sigma^{(n)})}\cD u((j-1)h,X^{(n)}_{(j-1)h})].
\end{equation}
The proof is similar to the proof of (\ref{3.7}), we just have to show that for
 every $\tau\in\cT^{(n)}$,
\begin{equation}\label{3.37}
P(\tau\wedge \sigma^{(n)},\mu \tau\wedge \sigma^{(n)}+X^{(n)}_{\tau\wedge 
\sigma^{(n)}})
\end{equation}
$$
\geq \psi(\mu \tau+X^{(n)}_{\tau})\bbI_{\{\tau\leq\sigma^{(n)}\}}+
\big(\delta-Ke(|\mu|h+2\kappa\sqrt{ h})\big)\bbI_{\{\sigma^{(n)}<\tau\}}.
$$
On the set $\tau\leq \sigma^{(n)}$ this inequality is clear since $P(t,x)
\geq \psi(x).$
For the case $\sigma^{(n)}<\tau$ observe that because $x>\ln K$ we must have 
$$
\ln K<\mu\sigma^{(n)}+X^{(n)}_{\sigma^{(n)}}<\ln K
+|\mu|h+2\kappa \sqrt{h}.
$$
By Theorem 8.1 in ~\cite{Ku} the right derivative $F_x(t,K+)$ at $K$ satisfies
$0>F_x(t,K+)>-1$ for any $t$, and so $0\leq F(t,K)-F(t,K+C\lambda)\leq  
C\lambda$ for each $C>0$ provided $0\leq\la\leq\la(C)$ is small enough.
Assume $0<\lambda<1$, then $e^{\lambda}-1 \leq \lambda e^{\lambda} \leq 
\lambda e$. Hence, taking $C=Ke$ we have
\begin{equation*}
P(t,\ln K)-P(t,\ln K+\lambda)=F(t,K)-F(t,Ke^{\lambda})\leq
 F(t,K)-F(t,K+Ke\lambda)\leq Ke\lambda.
\end{equation*}
Put $\lambda=|\mu|h+\kappa\sqrt{h}$ then for $\sigma^{(n)}<\tau$ and 
sufficiently large $n$, 
$$
P(\sigma^{(n)},\mu\sigma^{(n)}+X^{(n)}_{\sigma^{(n)}})\geq 
P(\sigma^{(n)},\ln K+\lambda)\geq\delta-Ke\lambda.
$$
 Hence, we obtain (\ref{3.37}) which yields also (\ref{3.35}).
To bound the right hand side of (\ref{3.36}) we split it similarly to the case  
$x\leq\ln K$ (see (\ref{3.8})) according to the three different regions 
$\textbf{C, B}$ and $\textbf{S}$. Since our process starts at $x>\ln K$, if 
$((j-1)h,X^{(n)}_{(j-1)h})\in \textbf{B}$ for some $j$ then 
this must happen after the time $\beta$, and so we can use (\ref{3.10}). 
The part that belongs to the region $\textbf{S}$ is non positive so we can 
ignore it, and so we will be left only with the region $\textbf{C}$.
\begin{lem}\label{lem 3.11}
For the discrete process $X^{(n)}_t$ such that $X^{(n)}_0=x>\ln K$ we have
$$
\bfE [\sum_{j=1}^{h^{-1}(\tau\wedge\sig^{(n)})}\cD |u((j-1)h,X_{(j-1)h}^{(n)})|
\bbI_{\{((j-1)h,X_{(j-1)h})\in\textbf{C}\}}]\leq Cn^{-3/4}.
$$
\end{lem}
\begin{proof}
It suffices to show that
\begin{equation}\label{inequality}
\bfE [\sum_{j=1}^{k_{\beta}\wedge(\sigma^{(n)}/h)}\cD |u((j-1)h,X_{(j-1)h}^{(n)})|
\bbI_{\{((j-1)h,X_{(j-1)h})\in\textbf{C}\}}]\leq Cn^{-3/4}
\end{equation} 
for some $C>0$ independent of $n$ 
%%%%%%%%%%%%%%%%%%% NEDD CHECK lem 3.2 below
since after time $\beta^{(n)}$ we come back to the American option case.
This is done in the same way as in Proposition \ref{prop 3.3}, and so we 
provide only a sketch of the proof. Let $c(s)=\ln K-\mu s+ \kappa\sqrt{h}$ 
 then similarly to the proof of Proposition \ref{prop 3.3} we obtain
\begin{equation}\label{3.38}
\sum_{j=0}^{k_{\beta}-1}E(|\cD u(jh,X^{(n)}_{jh})|
\bbI_{\{(jh,X^{(n)}_{jh})\in C \}\cap\{jh < \sigma^{(n)}\}})
\leq \frac{C_2}{n}+\frac{C_1}{n}\int_{h}^{\beta}\frac{ds}{\sqrt{s}}
\int_{ c(s)}^{\infty }dz|\frac{\partial^2 u }{\partial t^2}(s,z)|.
\end{equation}
Let $\frac{\sqrt{2}}{\kappa}(\ln K+|\mu|T)<k'$ and split the integral in 
(\ref{3.38}) into two parts
\begin{eqnarray}\label{3.39}
&\int_{h}^{\beta}\frac{ds}{\sqrt{s}}\int_{ c(s)}^{\infty }dz|\frac{\partial^2 
u }{\partial t^2}(s,z)|=\\
&\int_{h}^{\beta}\frac{ds}{\sqrt{s}}\int_{ c(s)}^{k'-\mu s }dz|\frac{\partial^2
 u }{\partial t^2}(s,z)|+\int_{h}^{\beta}\frac{ds}{\sqrt{s}}\int_{k'-
 \mu s}^{\infty }dz|\frac{\partial^2 u }{\partial t^2}(s,z)|.\nonumber
 \end{eqnarray}
Let $E=\{(s,z): 0<s<\beta,k'-\mu s <z<\infty\}$ then by Corollary 
\ref{cor_infty} we see that $\frac{\partial^2 u }{\partial t^2}(s,z)\in L^2[E]$, 
and so we obtain similarly to (\ref{3.18}) that for some constant $C>0$,
\begin{equation}\label{3.40}
\int_{h}^{\beta}\frac{ds}{\sqrt{s}}\int_{k'-\mu s}^{\infty }dz|
\frac{\partial^2 u }{\partial t^2}(s,z)|<C\ln n.
\end{equation}
In the first integral in the right hand side of (\ref{3.39}) we do the same 
procedure as in (\ref{3.15})-(\ref{lamstat3.5})  relying on Proposition 
\ref{thm2.5} and deriving that for some constant $C>0$,
\begin{equation}\label{3.41}
\int_{h}^{\beta}\frac{ds}{\sqrt{s}}\int_{ c(s)}^{k'-\mu s }dz|\frac{
\partial^2 u }{\partial t^2}(s,z)|<Cn^{1/4}.
\end{equation}
Combining (\ref{3.38})--(\ref{3.41}) we obtain (\ref{inequality}) and
 complete the proof of the lemma.
\end{proof}

An estimate for the lower bound of  $P^{(n)}_1(x)-P(0,x)$ when 
$x>\ln K$ is done similarly to the case $x\leq\ln K$. As in that case 
we use the stopping time $\tau^{(n)}$ from (\ref{3.23}) and from the
 above we see that (\ref{3.24}) is true also for the case under consideration.
 We consider again $\tau_\alpha^{(n)}$ defined before Lemma \ref{lem 3.8} and
  similarly to (\ref{3.25}) obtain that
\begin{eqnarray}\label{3.42}
&\big |\bfE [u(\tau_\alpha^{(n)}\wedge\sigma^{(n)},
X_{\tau_\alpha^{(n)}\wedge
\sigma^{(n)}})]\\
&-\bfE [e^{-r\tau_\alpha^{(n)}\wedge\sigma^{(n)}}\Big(\psi(\mu \tau_\alpha^{(n)}
+X^{(n)}_{\tau_\alpha^{(n)}})\bbI_{\{\tau_\alpha^{(n)}\leq\sigma^{(n)}\}}+
\big(\delta-Ke(|\mu|\sqrt{h}+\sigma{h})\big) \bbI_{\{\sigma^{(n)}
<\tau_\alpha^{(n)}\}} \big) \Big)] \big |.\nonumber
 \end{eqnarray}
In order to estimate (\ref{3.42}) for $x>\ln K$ we only need to split it 
into two parts, one for $\tau_\alpha^{(n)}\leq \sigma^{(n)}$ and the other one 
for $\sigma^{(n)}<\tau_\alpha^{(n)}$.
This it true in view of the fact that if we begin with $x>\ln K$ and 
$\tau_\alpha^{(n)}\leq \sigma^{(n)}$ then we must have $\beta^{(n)}\leq
\tau_\alpha^{(n)}$, and so we are back to the American option case and can 
use Lemma \ref{lem 3.8} for this case. If $\sigma^{(n)}<\tau_\alpha^{(n)}$ 
then the process $X^{(n)}$ is stopped near the seller's
 boundary and similarly to (\ref{3.32}) we can use the Lipschitz property 
 of $P$ to obtain,
$$
\bfE [\big(u(\sigma^{(n)},X^{(n)}_{\sigma^{(n)}})-e^{-r\sigma^{(n)}}(\delta
-Ke(|\mu|h+2\kappa\sqrt{h})) \bbI_{\{\sigma^{(n)}<\tau_\alpha^{(n)}\}}\big)]
\leq\frac{C}{\sqrt{n}}.
 $$
From here we can proceed similarly to the case of $x\leq \ln K$ and obtain the 
lower bound for $P^{(n)}_1$ proving (\ref{P1_err}) for $P_1^{(n)}$. \qed

Next, we turn to the second approximation function $P^{(n)}_2$, still in 
the case of $x>\ln K$. For the upper bound we use Lemma \ref{lem 3.10}
as in the case $x\leq \ln K$ and proceed similarly to the proof of the upper
bound for the first approximation function $P^{(n)}_1$.
The proof of the lower bound is similar to the case  $x\leq \ln K$ and 
 we obtain the result observing that if $x>\ln K$ then 
 $P(t,x)<\psi(x)+\delta=\delta$ for any $t\in[0,T]$. \qed

%%%%%%%%%%%%%%%%%%%%%%%%%%%%%%%%%%%%%%%%%%%%%%%%%%%%%%%%%%%%%%%%%%%%%%%%% 

\section{Computations}\label{sec6}
In this section we exhibit computations of price functions and free 
boundaries of game and American put options. All graphs of functions
related to game put options were plotted using the approximation function
   $P^{(2000)}_2$ (see (\ref{P^n2})). The graphs for the American put 
   options were computed using the approximation function $P_A^{(2000)}$
    from \cite{L}.
    
     Figure \ref{figure1}  shows both free boundaries of the holder and of the 
writer of a game put option and also the free boundary of the holder of an
American put option corresponding to the option parameters
$K =20,\ r = 0.02,\ \kappa =0.15,\ T=0.5,\ \delta = 0.15$.
Here $K$ is the strike of the option, $r$ is the interest rate, $\kappa$
 is the volatility, $T$ is the time to maturity and $\delta$ is the writer's
  cancelation penalty in the case of game option. 
  
  In Figure \ref{figure2} we plot the graphs of an American put option price
 function and of a game put option price functions with 
 $\delta = 1.0$ and $\delta = 1.5$ while other parameters are
$K =20,\ r = 0.02,\ \kappa =0.15,\ T=10$. To see what is what here we recall
that prices of game options do not exceed prices of corresponding American
options and higher penalties increase prices.

Figure \ref{figure3} shows the holder's free boundary for American and game
 put options where we use the same parameters as in Figure \ref{figure1} adding
 also plots of free boundaries for the game put options with penalty values
 $\delta = 0.3$ and $\delta = 0.5$.

\begin{figure}
    \centering
    \includegraphics[width=120mm,height=80mm]{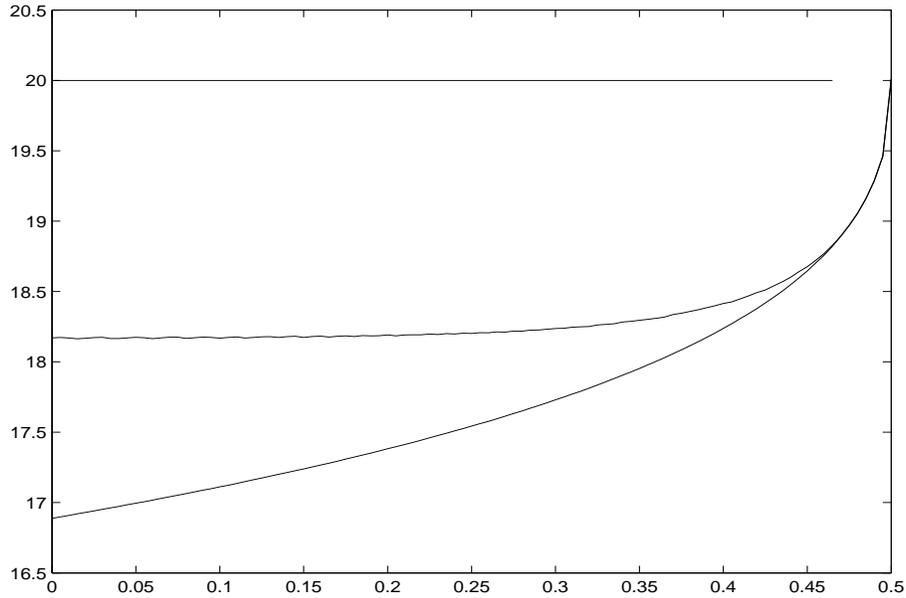}
    \caption{Free boundaries of American and game put options.}\label{figure1}
\end{figure}

\begin{figure}
    %\centering
    %\includegraphics[width=120mm,height=80mm]{AG_d_10_15.pdf}
    \includegraphics[width=120mm,height=80mm]{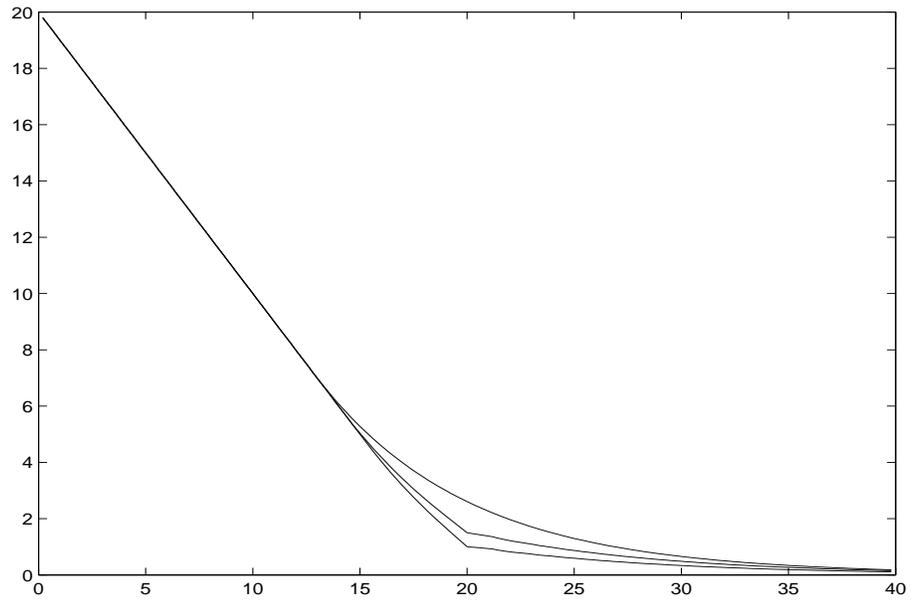}
    \caption{The price functions of American and game
     put options.}\label{figure2}
\end{figure}
\clearpage
\begin{figure}
    %\centering
    %\includegraphics[width=120mm,height=80mm]{op_B_Bun_G1_G3_G5_A.pdf}
    \includegraphics[width=120mm,height=80mm]{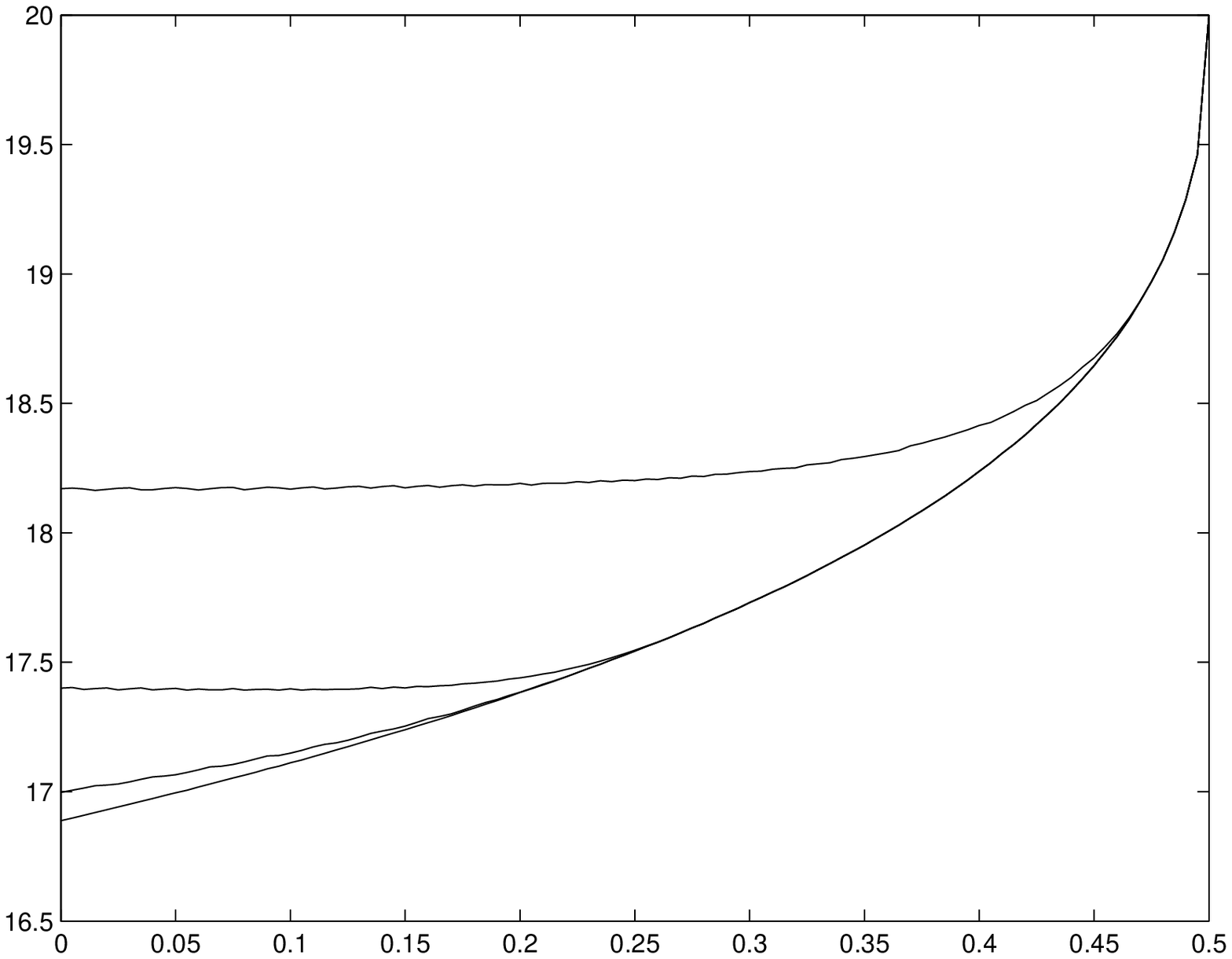}
    \caption{Holder free boundaries of American and game put 
    options.}\label{figure3}
\end{figure}

\end{document}